\documentclass[10pt, doublecolumn]{IEEEtran}
\usepackage[dvips]{graphicx}
\usepackage{graphics, subfigure, color, epsfig, psfrag, amsmath, amsfonts, amssymb, amsthm, eucal, balance}
\usepackage{bm}
\usepackage{url}%

\newtheorem{cor}{Corollary}
\newtheorem{lem}{Lemma}
\newtheorem{prop}{Proposition}

\begin{document}
\title{Mixed-ADC Massive MIMO Uplink in Frequency-Selective Channels}
\author{Ning Liang and Wenyi Zhang, \emph{Senior Member, IEEE}
\thanks{
N. Liang and W. Zhang are with Key Laboratory of Wireless-Optical Communications, Chinese Academy of Sciences, and Department of Electronic Engineering and Information Science, University of Science and Technology of China, Hefei, China (Emails: liangn@mail.ustc.edu.cn, wenyizha@ustc.edu.cn).

This work has been supported by the National High Technology Research and Development Program of China (863 Program) through grant 2014AA01A702, the National Natural Science Foundation of China under Grant 61379003, and the Fundamental Research Funds for the Central Universities under Grant WK3500000003.
}
}
\maketitle
\thispagestyle{empty}
\begin{abstract}
The aim of this paper is to investigate the recently developed mixed-ADC architecture for frequency-selective channels. Multi-carrier techniques such as orthogonal frequency division multiplexing (OFDM) are employed to handle inter-symbol interference (ISI). A frequency-domain equalizer is designed for mitigating the inter-carrier interference (ICI) introduced by the nonlinearity of one-bit quantization. For static single-input-multiple-output (SIMO) channels, a closed-form expression of the generalized mutual information (GMI) is derived, and based on which the linear frequency-domain equalizer is optimized. The analysis is then extended to ergodic time-varying SIMO channels with estimated channel state information (CSI), where numerically tight lower and upper bounds of the GMI are derived. The analytical framework is naturally applicable to the multi-user scenario, for both static and time-varying channels. Extensive numerical studies reveal that the mixed-ADC architecture with a small proportion of high-resolution ADCs does achieve a dominant portion of the achievable rate of ideal conventional architecture, and that it remarkably improves the performance as compared with one-bit massive MIMO.
\end{abstract}
\setcounter{page}{1}

\section{Introduction}
By deploying tens to hundreds of antennas at the base station (BS) and simultaneously serving multiple users in the same time-frequency resource block, massive multiple-input-multiple-output (MIMO) achieves unprecedented gain in both spectral efficiency and radiated energy efficiency, accommodating the stringent requirements of future 5G systems \cite{marzetta2010noncooperative}-\cite{ngo2013energy}. The performance gains, however, come at the expense of a linear increase in hardware cost as well as circuitry power consumption, and therefore massive MIMO will be more attractive if low-cost, energy-efficient solutions are available.
\subsection{Related Work}
Basically, if each BS antenna is configured with an unabridged radio frequency (RF) chain, then the only way to alleviate hardware cost and circuitry power consumption is to use economical low-power components when building the RF chains. These components, however, generally have to tolerate severe impairments, such as quantization noise, nonlinearity of power amplifier, phase noise of oscillator, and I/Q imbalance. By modeling the aggregate effect of the impairments (including quantization noise) as an additional Gaussian noise independent of the desired signal, the authors of \cite{bjornson2014hardware} investigated the impact of hardware impairments on the system spectral efficiency and radiated energy efficiency, and concluded that massive MIMO exhibits some degree of resilience against hardware impairments. Further, employing a similar model the authors of \cite{bjornson2015massive} derived a scaling law that reveals the tradeoff among hardware cost, circuitry power consumption, and the level of impairments. Although the adopted stochastic impairment models are not rigorous theoretically (for example, the quantization noise inherently depends on the desired signal), the analytical results in \cite{bjornson2014hardware}-\cite{bjornson2015massive} closely match those obtained by a more accurate hardware-specific deterministic model, as demonstrated by \cite{gustavsson2014impact}.

Among all the components in a RF chain, high-resolution ADC (typically with a bit-width exceeding 10) is particularly power-hungry, especially for wideband systems, since the power consumption of an ADC scales roughly exponentially with the bit-width and linearly with the baseband bandwidth \cite{murmann2014adc}. Lowering the bit-width of the adopted ADC will therefore bring in considerable savings on cost and energy. This fact actually has motivated extensive research on low-cost, energy-efficient design of wireless communication systems through employing low-resolution or even one-bit ADCs to build the RF chain; see, e.g., \cite{singh2009limits} for additive white Gaussian noise (AWGN) channels, \cite{yin2010monobit} for ultra-wideband channels, and \cite{mezghani2008analysis}-\cite{mo2015capacity} for MIMO channels.

Regarding massive MIMO, the impact of coarse quantization has been investigated only recently. In \cite{risi2014massive}, the authors evaluated the achievable rates of an uplink one-bit massive MIMO system adopting QPSK constellation, least-squares (LS) channel estimation, and maximum ratio combiner (MRC) or zero-forcing combiner (ZFC). The authors of \cite{jacobsson2015one}-\cite{jacobsson2016massive} further revealed that enhancement of achievable rates can be attained by high-order modulation such as 16-QAM. The underlying reason is that, even for one-bit massive MIMO, the amplitude of the transmit signal can still be recovered provided that the number of BS antennas is sufficiently large and that the signal-to-noise ratio (SNR) is not too high. Optimizations of pilot length and ADC bit-width were performed in \cite{fan2016optimal} and \cite{verenzuela2016hardware} respectively, both adopting MRC at the receiver. Recently, the authors of \cite{mollen2016performance} analyzed the achievable rates of one-bit massive MIMO in frequency-selective channels, employing linear minimum mean squared error (MMSE) channel estimator and linear combiners such as MRC and ZFC.

Beyond that, various channel estimation and data detection algorithms have been proposed for massive MIMO under coarse quantization. For example, near maximum likelihood (nML) detector and channel estimator were proposed in \cite{choi2016near} for one-bit massive MIMO in frequency-flat fading channels. In \cite{studer2016quantized}, channel estimation and data detection algorithms were developed for quantized massive MIMO in frequency-selective fading channels. Particularly, tradeoffs between error rate performance and computational complexity were investigated therein based on mismatched quantization models. Techniques based on message passing algorithm (and its variants) were also applied to quantized massive MIMO systems, such as \cite{wen2016bayes}-\cite{zhang2015mixed} for frequency-flat fading channels and \cite{wang2014multiuser}-\cite{wang2015multiuser} for frequency-selective fading channels. In general, \cite{risi2014massive}-\cite{wang2015multiuser} conclude that massive MIMO is somewhat robust to coarse quantization, validating the potential of building massive MIMO by low-resolution ADCs.\footnote{From an engineering perspective, coarse quantization actually subverts almost every aspect of the system design, including time-frequency synchronization, digital filtering, data detection, among others. In this paper, however, we primarily focus on the fundamental performance evaluation of such system, and leave the other practically important aspects for future research.}
\subsection{Mixed-ADC Architecture}
Except \cite{zhang2015mixed}, all the aforementioned works have assumed a homogeneous-ADC architecture; that is, all the antennas at the BS are equipped with low-resolution ADCs of the same bit-width. Although such an architecture seems feasible in terms of achievable rate or bit error rate (BER), it has several practical issues, including data rate loss in the high SNR regime \cite{mo2014high}-\cite{mo2015capacity}, error floor for linear multi-user  detection with 1-3 bit quantized outputs \cite{zhang2015mixed}-\cite{wang2014multiuser}, overhead and challenge of channel estimation \cite{risi2014massive}-\cite{mollen2016performance}, \cite{ivrlac2007mimo} and of time-frequency synchronization \cite{studer2016quantized} from quantized outputs. From this perspective, high-resolution ADCs can still be useful for effective design of massive MIMO receivers.

Motivated by such consideration, in early works \cite{liang2015mixed}-\cite{liang2015a} we have proposed a mixed-ADC architecture for massive MIMO, where a small proportion of the high-resolution ADCs are reserved while the others are replaced by one-bit ADCs.\footnote{Generally speaking, mixed-ADC architecture stands for any receiver architecture that contains ADCs of possibly different bit-widths, thus even including homogeneous-ADC architecture as a special case. Unless otherwise specified, however, the mixed-ADC architecture in this paper refers in particular to the one that is built upon one-bit and high-resolution ADCs, simply for analytical convenience.} For frequency-flat channels, \cite{liang2015mixed} shows that the mixed-ADC architecture is able to achieve an attractive tradeoff between spectral efficiency and energy efficiency. Moreover, compared with the homogeneous-ADC architecture, the mixed-ADC architecture is inherently immune to most of the aforementioned concerns. For example, channel estimation and time-frequency synchronization in the mixed-ADC architecture are more tractable than those in the homogeneous-ADC architecture \cite{sezginer2008asymptotically}, benefiting from the reserved high-resolution ADCs.

It is perhaps also worth noting that the mixed-ADC architecture is much more flexible to the time-varying property of the users' demand for mobile data traffic. To be specific, when the users' sum rate requirement is low, part of the BS antennas can be deactivated. Then high-resolution ADCs may be adopted in the channel training phase while one-bit ADCs may be employed in the data transmission phase. Compared with the homogeneous-ADC architecture, the mixed-ADC architecture in this situation incurs much lower channel estimation overhead and will therefore achieve higher energy efficiency.
\subsection{Contributions}
In this paper, we leverage the information-theoretical tool of generalized mutual information (GMI) to quantify the achievable rates of the mixed-ADC architecture in frequency-selective channels.\footnote{Due to the nonlinearity of coarse quantization, the frequency-selective channel cannot be decomposed into multiple independent frequency-flat subchannels by simply applying orthogonal frequency division multiplexing (OFDM). Therefore in this situation, channel estimation and data detection algorithms designed for frequency-flat channels are no longer applicable, and new methods have to be developed to handle the aforementioned issues.} The main contributions of this paper are summarized as follows:
\begin{itemize}
  \item We modify the mixed-ADC architecture to make it suitable for frequency-selective channels, adopting OFDM to handle inter-symbol interference (ISI) and a linear frequency-domain equalizer to mitigate inter-carrier interference (ICI).
  \item For static SIMO channels, we derive an explicit expression of the GMI, and based on which further optimize the linear frequency-domain equalizer. The analytical results are then extended to ergodic time-varying SIMO channels, where tight lower and upper bounds of the GMI are derived. The impact of frequency diversity and imperfect CSI on the system performance is investigated as well.
  \item We then extend the analytical framework to the multi-user scenario. BER performance is also examined for a practical convolutional codec.
  \item We develop a reduced-complexity algorithm, by which the computational complexity of the linear frequency-domain equalizer is reduced from $O(N^3Q^3)$ to $O(\max\{N^3Q,N^2Q^2\log_2 Q\})$, where $N$ is the number of BS antennas and $Q$ is the number of subcarriers.
\end{itemize}

Extensive numerical studies under various setups reveal that, with only a small proportion of high-resolution ADCs, the mixed-ADC architecture attains a large portion of the achievable rate of ideal conventional architecture, and significantly outperforms antenna selection with the same number of high-resolution ADCs. In addition, the mixed-ADC architecture in the multi-user scenario remarkably lowers the error floor encountered by one-bit massive MIMO. These observations validate the merits of the mixed-ADC architecture for effective design of massive MIMO receivers.
\subsection{Notation}
Throughout this paper, vectors and matrices are given in bold typeface, e.g., $\mathbf{x}$ and $\mathbf{X}$, respectively, while scalars are given in regular typeface, e.g., $x$. We let $\mathbf{X}^{*}$, $\mathbf{X}^{t}$ and $\mathbf{X}^{\dag}$ denote the conjugate, transpose and conjugate transpose of $\mathbf{X}$, respectively. The $q$-th element of vector $\mathbf{x}$ is symbolized as $(\mathbf{x})_q$, and in the meantime, the $(p,q)$-th element of matrix $\mathbf{X}$ is symbolized as $(\mathbf{X})_{pq}$. Notation $\mathrm{diag}(\cdot)$ denotes a diagonal matrix, with the diagonal elements numerated in the bracket. For a positive integer $N$, we use $\mathbb{N}$ to represent the set of positive integers no larger than $N$, i.e., $\mathbb{N}=\{1,...,N\}$. For a positive real number $x$, we use $\lceil{x}\rceil$ to denote the minimum integer that satisfies $\lceil{x}\rceil\geq x$. Notation $\mathcal{CN}(\bm{\mu},\mathbf{C})$ stands for the distribution of a circularly symmetric complex Gaussian random vector with mean vector $\bm{{\mu}}$ and covariance matrix $\mathbf{C}$. Subscripts $\mathrm{R}$ and $\mathrm{I}$ are used to indicate the real and imaginary parts of a complex number, respectively, e.g., $x=x_{\mathrm{R}}+j x_{\mathrm{I}}$, with $j$ being the imaginary unit. We further use $\mathrm{sgn}(x)=\frac{1}{\sqrt{2}}[\mathrm{sgn}(x_{\mathrm{R}})+j\mathrm{sgn}(x_{\mathrm{I}})]$ to denote the sign function of a complex number $x$, and  $\log(x)$ to denote the natural logarithm of a positive real number $x$.
\subsection{Outline}
The remaining part of this paper is organized as follows. Section \ref{sect:model} describes the system model in the single-user scenario. For static SIMO channels, Section \ref{sect:single-user} first derives an explicit expression of the GMI and then optimizes the linear frequency-domain equalizer. Besides, properties of the GMI in several special cases are explored and the analytical results are further extended to ergodic time-varying SIMO channels. Section \ref{sect:multi-user} applies the analytical framework to the multi-user scenario. A reduced-complexity algorithm is proposed in Section \ref{sect:complexity} for efficiently implementing the linear frequency-domain equalizer. Numerical results are presented in Section \ref{sect:numerical} to corroborate the analysis. Finally, Section \ref{sect:conclusion} concludes this paper. Auxiliary technical derivations are collected in the appendix.

\begin{figure*}
  \centering
  \includegraphics[width=0.9 \textwidth]{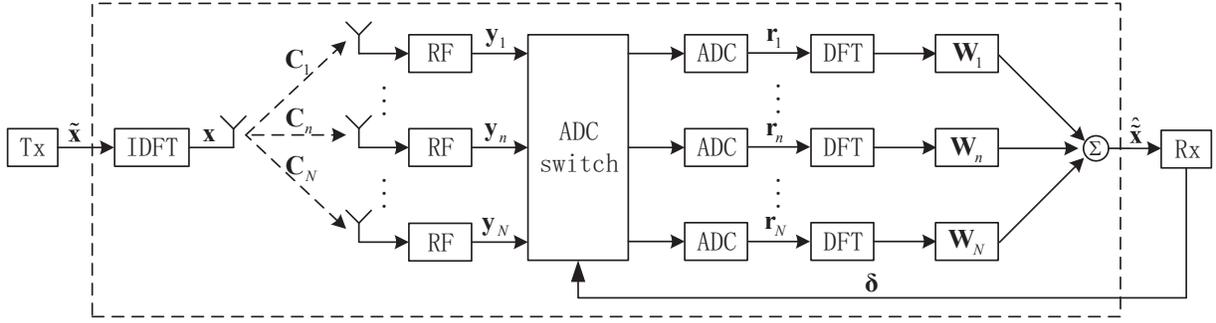}
  \caption{Illustration of the mixed-ADC architecture for frequency-selective SIMO channels. We note that the ADC switch module can also be placed before the RF chains. In that manner, the RF chain followed by a pair of one-bit ADCs can be manufactured with lower quality requirements and consequently we can further reduce the power consumption and hardware cost \cite{studer2016quantized}. On the other hand, switching at radio frequency may be more challenging and costly than at baseband. Which choice is favorable will be determined by practical engineering.}
  \label{fig:system model}
\end{figure*}
\section{System model}
\label{sect:model}
A single-antenna user communicates to an $N$-antenna BS through a frequency-selective SIMO channel, of which each branch consists of $T$ taps. We start by focusing on static channels and assuming perfect CSI at the BS. Particularly, OFDM is adopted to handle ISI. We denote the frequency-domain OFDM symbol by $\tilde{\mathbf{x}}\in\mathbb{C}^{Q\times 1}$ and its time-domain counterpart by $\mathbf{x}=\mathbf{F}^{\dag}\tilde{\mathbf{x}}$, where $Q$ is the number of subcarriers and the discrete Fourier transform (DFT) matrix $\mathbf{F}$ satisfies $\mathbf{F}\mathbf{F}^{\dag}=\mathbf{F}^{\dag}\mathbf{F}=\mathbf{I}_{Q}$.

For the branch related to the $n$-th BS antenna, we denote it by $\mathbf{h}_n\triangleq[h_{n1},...,h_{nT},0,...,0]^t\in\mathbb{C}^{Q\times 1}$, and accordingly, its circulant matrix form by $\mathbf{C}_n$. Note that $\mathbf{C}_n$ can be decomposed as $\mathbf{C}_n=\mathbf{F}^{\dag}\bm{\Lambda}_n\mathbf{F}$, with the diagonal matrix $\bm{\Lambda}_n\triangleq\mathrm{diag}(\bm{\lambda}_n)$ given by
\begin{equation}
\bm{\lambda}_n=\sqrt{Q}\mathbf{F}\mathbf{h}_n.
\end{equation}
The channel outputs at the $n$-th BS antenna over $Q$ channel uses (with cyclic prefix removed) can be collectively written as
\begin{equation}
\mathbf{y}_n=\mathbf{C}_n\mathbf{x}+\mathbf{z}_n,
\end{equation}
where $\mathbf{z}_n\sim\mathcal{CN}(\mathbf{0},\sigma^2\mathbf{I}_{Q})$ collects the independent and identically distributed (i.i.d.) complex Gaussian noise.

To fulfill signal processing in the digital domain, $\mathbf{y}_n$ needs to be quantized by a pair of ADCs, one for each of the real and imaginary parts. For the mixed-ADC architecture, there are only $K$ pairs of high-resolution ADCs available at the BS and all the other $(N-K)$ pairs of ADCs are with only one-bit resolution. Thus the quantized output can be expressed as
\begin{equation}
\mathbf{r}_n=\delta_n\mathbf{y}_n+\bar{\delta}_n\mathrm{sgn}(\mathbf{y}_n).
\end{equation}
Here, $\delta_n\in\{0,1\}$, $\bar{\delta}_n\triangleq 1-\delta_n$, and $\sum_{n=1}^N \delta_n=K$. Particularly, $\delta_n=1$ means that $\mathbf{y}_n$ is quantized by a pair of high-resolution ADCs, whereas $\delta_n=0$ indicates that $\mathbf{y}_n$ is quantized by a pair of one-bit ADCs. We further define an ADC switch vector $\bm{\delta}\triangleq[\delta_1,...,\delta_N]^t$, which should be optimized according to the channel state $\{\mathbf{h}_n\}_{n=1}^N$ to maximize the user's data rate.

Then we employ a DFT to $\mathbf{r}_n$. Due to the strong nonlinearity of one-bit quantization, severe ICI is introduced during the time-frequency conversion. To handle this, we propose a linear frequency-domain equalizer as illustrated in Figure \ref{fig:system model}. Accordingly, the processed output is
\begin{equation}
\hat{\tilde{\mathbf{x}}}=\sum_{n=1}^N \mathbf{W}_n\mathbf{F}\mathbf{r}_n,
\end{equation}
where $\mathbf{W}_n\triangleq\mathrm{diag}(\mathbf{w}_n)$ is a $Q$-dimensional diagonal matrix.\footnote{In this paper, we restrict $\mathbf{W}_n$ to be diagonal, and for analytical convenience, the weight for each BS antenna is absorbed into $\mathbf{W}_n$.} For expositional concision, we may define $\mathbf{w}\triangleq[\mathbf{w}_1^t,...,\mathbf{w}_N^t]^t$, and it should be optimized according to $\{\mathbf{h}_n\}_{n=1}^N$ and $\bm{\delta}$ to maximize the user's achievable rate.

For analytical convenience, we let the decoder adopt a generalized nearest-neighbor decoding rule; that is, upon observing $\{\hat{\tilde{\mathbf{x}}}[l]\}_{l=1}^{L}$, it computes, for each possible input message $m$, the distance metric\footnote{For tractability consideration, here the scaling parameter for each subcarrier is designated an identical value $a$. Although such a choice is generally suboptimal, we note that the resulting performance loss is supposed to be marginal, benefiting from the channel hardening effect of massive MIMO.}
\begin{equation}
D(m)=\frac{1}{L}\sum_{l=1}^L\|\hat{\tilde{\mathbf{x}}}[l]-a\tilde{\mathbf{x}}[m,l]\|^2,\ \ m\in\mathcal{M},
\label{equ:equ_14}
\end{equation}
and determines $\hat{m}$ as the one that minimizes $D(m)$. Here $\mathcal{M}$ denotes the set of all the possible messages, $\{\tilde{\mathbf{x}}[m,l]\}_{l=1}^L$ denotes the codeword for message $m$ in the frequency domain, and $L$ is the codeword length measured in OFDM symbol. We restrict the codebook to be drawn from a Gaussian ensemble; that is, each codeword is a sequence of $L$ i.i.d. $\mathcal{CN}(\mathbf{0},\mathcal{E}_{\mathrm{s}}\mathbf{I}_Q)$ random vectors, and all the codewords are mutually independent. Such a choice of the codebook ensemble satisfies an average power constraint of $\mathcal{E}_{\mathrm{s}}$, and therefore we define SNR as $\mathrm{SNR}\triangleq\mathcal{E}_{\mathrm{s}}/\sigma^2$, hereafter letting $\sigma^2=1$ for concision.
\section{Performance in Single-User Scenario}
\label{sect:single-user}
For the mixed-ADC architecture, since there is no closed-form expression of the channel capacity, in this paper, we leverage GMI to evaluate its achievable rates. The GMI is a lower bound of the channel capacity, and more precisely, it characterizes the maximum achievable rate of specific i.i.d. random codebook ensemble (Gaussian ensemble here) and specific decoding rule (generalized nearest-neighbor decoding here) such that the average decoding error probability (averaged over the codebook ensemble) is guaranteed to vanish asymptotically as the codeword length $L$ grows without bound \cite{lapidoth2002fading}. As a performance metric, it has proven convenient and useful in several important scenarios, such as fading channel with imperfect CSI at the receiver \cite{lapidoth2002fading} and channels with transceiver distortion \cite{liang2015mixed}-\cite{liang2015a}, \cite{zhang2012general}-\cite{vehkapera2015asymptotic}.
\subsection{GMI and Optimal Linear Frequency-Domain Equalizer}
We exploit the theoretical framework in \cite{liang2015mixed} and \cite{zhang2012general} to derive the GMI of the mixed-ADC architecture. Following essentially the same steps as \cite[App. C]{zhang2012general}, we obtain an explicit expression of the GMI as follows.

\begin{prop}
\label{prop:prop_3}
Assuming Gaussian codebook ensemble and generalized nearest-neighbor decoding, the GMI for given $\mathbf{w}$ and $\bm{\delta}$ is
\begin{equation}
I_{\mathrm{GMI}}(\mathbf{w},\bm{\delta})=\log\left(1+\frac{\Delta(\mathbf{w},\bm{\delta})}{1-\Delta(\mathbf{w},\bm{\delta})}\right),
\label{equ:equ_16}
\end{equation}
where the performance indicator $\Delta(\mathbf{w},\bm{\delta})$ follows from
\begin{equation}
\Delta(\mathbf{w},\bm{\delta})=\frac{|\mathbb{E}[\hat{\tilde{\mathbf{x}}}^{\dag}\tilde{\mathbf{x}}]|^2}
{Q\mathcal{E}_s\mathbb{E}[\hat{\tilde{\mathbf{x}}}^{\dag}\hat{\tilde{\mathbf{x}}}]}.
\label{equ:equ_17}
\end{equation}
The corresponding optimal scaling parameter is given by
\begin{equation}
a_{\mathrm{opt}}(\mathbf{w},\bm{\delta})=\frac{1}{Q\mathcal{E}_{\mathrm{s}}}\mathbb{E}[\tilde{\mathbf{x}}^{\dag}\hat{\tilde{\mathbf{x}}}].
\end{equation}
We note that here the expectation operation $\mathbb{E}[\cdot]$ is taken with respect to $\tilde{\mathbf{x}}$ and $\mathbf{z}_n,n\in\mathbb{N}$.
\end{prop}

We notice that $\Delta(\mathbf{w},\bm{\delta})$ is lower bounded by zero, and that from Cauchy-Schwartz's inequality, it is upper bounded by one. Moreover, $I_{\mathrm{GMI}}(\mathbf{w},\bm{\delta})$ is a strictly increasing function of $\Delta(\mathbf{w},\bm{\delta})$ for $\Delta(\mathbf{w},\bm{\delta})\in[0,1)$, and thus in the following, we only need to maximize $\Delta(\mathbf{w},\bm{\delta})$ by optimizing the design of $\mathbf{w}$ and $\bm{\delta}$. The optimal linear frequency-domain equalizer is given by the proposition below.

\begin{prop}
\label{prop:prop_1}
For given $\{\mathbf{h}_n\}_{n=1}^N$ and $\bm{\delta}$, the optimal linear frequency-domain equalizer $\mathbf{w}_{\mathrm{opt}}$ is
\begin{equation}
\mathbf{w}_{\mathrm{opt}}=\mathbf{D}^{-1}\mathbf{g}.
\label{equ:prop_1.1}
\end{equation}
Accordingly, the maximized $\Delta(\mathbf{w},\bm{\delta})$ and the optimal choice of $a_{\mathrm{opt}}(\mathbf{w},\bm{\delta})$ are given by
\begin{equation}
\Delta(\mathbf{w}_{\mathrm{opt}},\bm{\delta})=a_{\mathrm{opt}}(\mathbf{w}_{\mathrm{opt}},\bm{\delta})=\frac{1}{Q\mathcal{E}_{\mathrm{s}}}\mathbf{g}^{\dag}\mathbf{D}^{-1}\mathbf{g}.
\label{equ:prop_1.2}
\end{equation}
In both equations, $\mathbf{g}\triangleq[\mathbf{g}_1^t,...,\mathbf{g}_N^t]^t\in\mathbb{C}^{NQ\times 1}$, of which the $n$-th segment $\mathbf{g}_n$ is given by
\begin{equation}
\mathbf{g}_n=\delta_n\cdot\mathcal{E}_{\mathrm{s}}\bm{\lambda}_n^*+\bar{\delta}_n\cdot\sqrt{\frac{2}{\pi}}\frac{\mathcal{E}_{\mathrm{s}}\bm{\lambda}_n^*}{\sqrt{1+\mathcal{E}_{\mathrm{s}}\|\bm{\lambda}_n\|^2/Q}}.
\label{equ:prop_1.3}
\end{equation}
In the meantime, $\mathbf{D}\in\mathbb{C}^{NQ\times NQ}$ is a block matrix
\begin{equation}
\mathbf{D}\triangleq
\begin{pmatrix}
    \mathbf{D}_{11} & \mathbf{D}_{21} & \cdots & \mathbf{D}_{N1} \\
    \mathbf{D}_{12} & \mathbf{D}_{22} & \cdots & \mathbf{D}_{N2} \\
    \vdots  & \vdots  & \ddots & \vdots  \\
    \mathbf{D}_{1N} & \mathbf{D}_{2N} & \cdots & \mathbf{D}_{NN}
\end{pmatrix},
\label{equ:prop_1.4}
\end{equation}
in which each block $\mathbf{D}_{nm}$ is a $Q$-dimensional diagonal matrix defined as
\begin{equation}
(\mathbf{D}_{nm})_{qq}=(\mathbf{F}\mathbf{R}_{nm}\mathbf{F}^{\dag})_{qq}.
\label{equ:prop_1.5}
\end{equation}
Here $\mathbf{R}_{nm}\triangleq\mathbb{E}[\mathbf{r}_n\mathbf{r}_m^{\dag}]$ is analytically evaluated in Appendix-A, and the expectation operation $\mathbb{E}[\cdot]$ is taken with respect to $\tilde{\mathbf{x}}$ and $\mathbf{z}_n,n\in\mathbb{N}$.
\end{prop}
\begin{proof}
For the readers' better understanding, here we outline a sketch of the proof. For a complete version of this proof, please refer to Appendix-A.

First, to maximize $\Delta(\mathbf{w},\bm{\delta})$ we need to derive the closed-form expressions of $\mathbb{E}[\hat{\tilde{\mathbf{x}}}^{\dag}\tilde{\mathbf{x}}]$ and $\mathbb{E}[\hat{\tilde{\mathbf{x}}}^{\dag}\hat{\tilde{\mathbf{x}}}]$. Through tedious manipulations, we obtain
\begin{equation}
\mathbb{E}[\hat{\tilde{\mathbf{x}}}^{\dag}\tilde{\mathbf{x}}]=\mathbf{w}^{\dag}\mathbf{g},\ \text{and}\
\mathbb{E}[\hat{\tilde{\mathbf{x}}}^{\dag}\hat{\tilde{\mathbf{x}}}]=\mathbf{w}^{\dag}\mathbf{D}\mathbf{w}.
\end{equation}
Then, $\Delta(\mathbf{w},\bm{\delta})$ in \eqref{equ:equ_17} yields a closed-form expression as
\begin{equation}
\Delta(\mathbf{w},\bm{\delta})=\frac{\mathbf{w}^{\dag}\mathbf{g}\mathbf{g}^{\dag}\mathbf{w}}{Q\mathcal{E}_{\mathrm{s}}\mathbf{w}^{\dag}\mathbf{D}\mathbf{w}}.
\end{equation}
We notice that it is actually a generalized Rayleigh quotient of $\mathbf{w}$, and as a result, we conveniently obtain the optimal linear frequency-domain equalizer as given by \eqref{equ:prop_1.1}.
\end{proof}
\subsection{Two Corollaries from Proposition \ref{prop:prop_1}}
In the previous subsection, we derived the optimal linear frequency-domain equalizer. Thus we are now ready to examine the performance of the proposed mixed-ADC architecture. Particularly in this subsection, we focus on two special case studies. For the special case of $\bm{\delta}=\mathbf{1}_Q$, the next corollary gives a comparison between the GMI and the channel capacity.
\begin{cor}
\label{cor:cor_1}
When $\bm{\delta}=\mathbf{1}_Q$, $I_{\mathrm{GMI}}(\mathbf{w}_{\mathrm{opt}},\bm{\delta})$ yields a simplified expression as
\begin{equation}
I_{\mathrm{GMI}}(\mathbf{w}_{\mathrm{opt}},\bm{\delta})=-\log\left(\frac{1}{Q}\sum_{q=1}^Q\frac{1}{1+\mathcal{E}_{\mathrm{s}}\sum_{n=1}^N|\lambda_{nq}|^2}\right).
\label{equ:equ_7}
\end{equation}
On the other hand, the channel capacity (achieved by MRC over the $N$ antennas at each subcarrier) is given by
\begin{equation}
C=-\frac{1}{Q}\sum_{q=1}^Q\log\left(\frac{1}{1+\mathcal{E}_{\mathrm{s}}\sum_{n=1}^N|\lambda_{nq}|^2}\right).
\end{equation}
\end{cor}

The proof is given in Appendix-B. Since $-\log(x)$ is a convex function of positive real number $x$, $I_{\mathrm{GMI}}(\mathbf{w}_{\mathrm{opt}},\bm{\delta})\leq C$ holds even when $\bm{\delta}=\mathbf{1}_Q$. We note that this rate loss is due to the identical choice of the scaling parameter $a_{\mathrm{opt}}(\mathbf{w}_{\mathrm{opt}},\bm{\delta})$ over all the subcarriers in the decoding metric \eqref{equ:equ_14}, and that benefiting from the channel hardening effect of massive MIMO, such performance loss is expected to be marginal, as will be verified by the numerical study in Section \ref{sect:numerical}.

The next corollary draws some connection between the frequency-flat channel scenario addressed in \cite{liang2015mixed}-\cite{liang2015a} and the frequency-selective channel scenario investigated in this paper.
\begin{cor}
\label{cor:cor_2}
When $T=1$, the analytical results in Proposition 2 reduce to those for frequency-flat channels obtained in \cite{liang2015mixed}-\cite{liang2015a}.
\end{cor}

See Appendix-C for its proof. The impact of frequency diversity on the system performance will be revealed by numerical studies in Section \ref{sect:numerical}.
\subsection{Asymptotic Behavior of $I_{\mathrm{GMI}}(\mathbf{w}_{\mathrm{opt}},\bm{\delta})$ in Low or High SNR Regime}
In the following, we first examine the asymptotic behavior of $I_{\mathrm{GMI}}(\mathbf{w}_{\mathrm{opt}},\bm{\delta})$ in the low SNR regime, and then for the special case of $\bm{\delta}=\mathbf{0}$, explore the limit of $I_{\mathrm{GMI}}(\mathbf{w}_{\mathrm{opt}},\bm{\delta})$ in the high SNR regime.
\begin{cor}
\label{cor:cor_3}
As $\mathcal{E}_{\mathrm{s}}\rightarrow 0$, for a given ADC switch vector $\bm{\delta}$ we have
\begin{equation}
I_{\mathrm{GMI}}(\mathbf{w}_{\mathrm{opt}},\bm{\delta})=\frac{1}{Q}\sum_{n=1}^N\left(\delta_n+\bar{\delta}_n\cdot\frac{2}{\pi}\right)\|\bm{\lambda}_n\|^2\mathcal{E}_{\mathrm{s}}+o(\mathcal{E}_{\mathrm{s}}).
\label{equ:equ_12}
\end{equation}
On the other hand, the channel capacity for $\bm{\delta}=\mathbf{1}_Q$ approaches
\begin{equation}
C=\frac{1}{Q}\sum_{n=1}^N\|\bm{\lambda}_n\|^2\mathcal{E}_{\mathrm{s}}+o(\mathcal{E}_{\mathrm{s}}).
\label{equ:equ_18}
\end{equation}
\end{cor}

The proof is given in Appendix-D. Comparing \eqref{equ:equ_12} and \eqref{equ:equ_18}, we can make two observations. First, due to one-bit quantization, part of the achievable data rate is degraded by a factor $\frac{2}{\pi}$. Second, to achieve the maximum data rate, high-resolution ADCs should be switched to the antennas with the maximum $\|\bm{\lambda}_n\|^2$ or $\|\mathbf{h}_n\|^2$, equivalently.

For a general $\bm{\delta}$, $\mathbf{D}^{-1}$ in the high SNR regime is still too complicated to yield any simplification. In light of this, we derive the limit of $\Delta(\mathbf{w}_{\mathrm{opt}},\bm{\delta})$ in the high SNR regime for the special case of $\bm{\delta}=\mathbf{0}$. The subsequent corollary gives our result.
\begin{cor}
\label{cor:cor_4}
For the special case of $\bm{\delta}=\mathbf{0}$, $\Delta(\mathbf{w}_{\mathrm{opt}},\bm{\delta})$ in the high SNR regime is given by the following limit
\begin{equation}
\lim_{\mathcal{E}_{\mathrm{s}}\rightarrow\infty}\Delta(\mathbf{w}_{\mathrm{opt}},\bm{\delta})=\frac{2}{\pi}\bar{\bm{\lambda}}^t\bar{\mathbf{D}}^{-1}\bar{\bm{\lambda}}^*,
\label{equ:equ_19}
\end{equation}
where $\bar{\bm{\lambda}}\triangleq[\bm{\lambda}_1^t/\|\bm{\lambda}_1\|,...,\bm{\lambda}_N^t/\|\bm{\lambda}_N\|]^t$, and $\bar{\mathbf{D}}\triangleq\lim_{\mathcal{E}_{\mathrm{s}}\rightarrow\infty}\mathbf{D}$.
\end{cor}

The proof is given in Appendix-E. Since both $\bar{\bm{\lambda}}$ and $\bar{\mathbf{D}}$ do not depend on $\mathcal{E}_{\mathrm{s}}$, \eqref{equ:equ_19} implies that, if all the BS antennas are equipped with one-bit ADCs, the GMI will ultimately approach a positive constant, as $\mathcal{E}_{\mathrm{s}}$ grows large. On the other hand, even if there are only one pair of high-resolution ADCs available at the BS, the GMI will always increase unboundedly with increasing $\mathcal{E}_{\mathrm{s}}$.\footnote{From \eqref{equ:equ_7} it can be readily inferred that the GMI of antenna selection with $K\geq 1$ will increase unboundedly with increasing $\mathcal{E}_{\mathrm{s}}$. Moreover, it is easy to verify that the mixed-ADC architecture achieves better performance than antenna selection with the same number of high-resolution ADCs. } From this perspective, the contribution of antennas connected with one-bit ADCs may be negligible in the high-SNR regime, and this in turn suggests us, intuitively, to switch the high-resolution ADCs to the antennas with the maximum $\|\bm{\lambda}_n\|^2$ or $\|\mathbf{h}_n\|^2$, equivalently.
\subsection{Extension to Ergodic Time-Varying Channels}
\label{subsect:ergodic_single}
Although our analysis thus far has been dedicated to the static channel scenario, the analytical framework developed can be extended to the time-varying channel scenario. We assume that the time-varying channel under consideration obeys the block fading model. Within each channel coherence interval, CSI needs to be explicitly or implicitly acquired, e.g., via channel training. Since the channel estimation from one-bit quantized outputs is inefficient and is elusive for analysis, in this paper we only activate the high-resolution ADCs in the channel training phase.

Particularly, channel training is performed in a round-robin manner. That is, in the first OFDM symbol interval, we switch the $K$ pairs of high-resolution ADCs to the first $K$ antennas and estimate $\mathbf{h}_1,...,\mathbf{h}_K$; in the subsequent OFDM symbol interval, we switch the high-resolution ADCs to the next $K$ antennas and estimate $\mathbf{h}_{K+1},...,\mathbf{h}_{2K}$; and so on. As a result, the training procedure consumes $\lceil N/K\rceil$ OFDM symbol intervals. To simplify analysis, in this subsection we assume that the channel coefficients follow i.i.d. Rayleigh fading, i.e., $h_{nt}\sim\mathcal{CN}(0,1/T)$, for any $n\in\mathbb{N}$ and $t\in\mathbb{T}$. The BS adopts an MMSE estimator, and thus without loss of generality, we can decompose $\mathbf{h}_n$ into
\begin{equation}
\mathbf{h}_n=\hat{\mathbf{h}}_n+\tilde{\mathbf{h}}_n,
\end{equation}
where $\hat{\mathbf{h}}_n$ is the estimated channel vector and $\tilde{\mathbf{h}}_n$ is the independent error vector. Both $\hat{\mathbf{h}}_n$ and $\tilde{\mathbf{h}}_n$ are complex Gaussian distributed. Moreover, if we define the normalized MSE as $\mathrm{MSE}_h=\sigma_h^2$, then $\hat{h}_{nt}\sim\mathcal{CN}(0,(1-\sigma_h^2)/T)$ and $\tilde{h}_{nt}\sim\mathcal{CN}(0,\sigma_h^2/T)$, for any $n\in\mathbb{N}$ and $t\in\mathbb{T}$.

Following a similar technical route as \cite[Prop. 4 \& 5]{liang2015mixed}, we obtain lower and upper bounds of the GMI, summarized by the proposition below.
\begin{prop}
For block fading channels with imperfect CSI, a lower bound of the GMI is
\begin{equation}
I_{\mathrm{GMI}}^{\mathrm{lower}}=\rho\log\left(1+\frac{\mathbb{E}_{\hat{\mathbf{h}}_n,n\in\mathbb{N}}[\Delta(\mathbf{w}_{\mathrm{opt}}^{\mathrm{im}},\bm{\delta})]}{1-\mathbb{E}_{\hat{\mathbf{h}}_n,n\in\mathbb{N}}[\Delta(\mathbf{w}_{\mathrm{opt}}^{\mathrm{im}},\bm{\delta})]}\right),
\end{equation}
and an upper bound of the GMI is
\begin{equation}
I_{\mathrm{GMI}}^{\mathrm{upper}}=\rho\mathbb{E}_{\hat{\mathbf{h}}_n,n\in\mathbb{N}}\left[\log\left(1+\frac{\Delta(\mathbf{w}_{\mathrm{opt}}^{\mathrm{im}},\bm{\delta})}{1-\Delta(\mathbf{w}_{\mathrm{opt}}^{\mathrm{im}},\bm{\delta})}\right)\right].
\end{equation}
Here $\rho\triangleq\frac{T_{\mathrm{c}}-\lceil N/K\rceil}{T_{\mathrm{c}}}$ accounts for the overhead of channel training, with $T_{\mathrm{c}}$ being the coherence interval length. Moreover, $\mathbf{w}_{\mathrm{opt}}^{\mathrm{im}}$ and $\Delta(\mathbf{w}_{\mathrm{opt}}^{\mathrm{im}},\bm{\delta})$ also come from \eqref{equ:prop_1.1} and \eqref{equ:prop_1.2}, but therein we need to replace $\mathbf{g}$ with $\mathbb{E}_{{\tilde{\mathbf{h}}_n,n\in\mathbb{N}}}[\mathbf{g}|\hat{\mathbf{h}}_n,n\in\mathbb{N}]$, and replace $\mathbf{D}$ with $\mathbb{E}_{{\tilde{\mathbf{h}}_n,n\in\mathbb{N}}}[\mathbf{D}|\hat{\mathbf{h}}_n,n\in\mathbb{N}]$.
\end{prop}

Due to space limitation, the proof is omitted. When $\rho=1$ and $\sigma_h^2=0$, the above results apply directly to time-varying channels with perfect CSI at the BS. Numerical results will be given in Section \ref{sect:numerical} to confirm the tightness of the bounds.

\section{Multi-user scenario}
\label{sect:multi-user}
\subsection{System Description}
A total of $U$ single-antenna users communicate to the BS simultaneously. Again we start from the static channel scenario with perfect CSI at the BS. The $u$-th user's frequency-domain OFDM symbol is $\mathcal{CN}(\mathbf{0},\mathcal{E}_{\mathrm{s}}\mathbf{I}_Q)$ distributed and is denoted by $\tilde{\mathbf{x}}^u$. Before transmission, an IDFT is applied to $\tilde{\mathbf{x}}^u$, and thus the time-domain transmit symbol of user $u$ is $\mathbf{x}^u=\mathbf{F}^{\dag}\tilde{\mathbf{x}}^u$, which is also $\mathcal{CN}(\mathbf{0},\mathcal{E}_{\mathrm{s}}\mathbf{I}_Q)$ distributed. We further assume that the OFDM symbols from different users are independent, since typically the users do not cooperate with each other.

We denote the multipath channel between user $u$ and the $n$-th BS antenna by $\mathbf{h}_n^u\triangleq[h_{n1}^u,...,h_{nT^u}^u,0,...,0]^t\in\mathbb{C}^{Q\times 1}$ and further its circulant matrix form by $\mathbf{C}_n^u$. Then, $\mathbf{C}_n^u$ can be decomposed as $\mathbf{C}_n^u=\mathbf{F}^{\dag}\bm{\Lambda}_n^u\mathbf{F}$, where the eigenvalue matrix $\bm{\Lambda}_n^u\triangleq\mathrm{diag}(\bm{\lambda}_{n})$ is given by
\begin{equation}
\bm{\lambda}_n^u=\sqrt{Q}\mathbf{F}\mathbf{h}_n^u.
\end{equation}
Note that we allow $T^u$ to be different for different users since they may be located in various environments. Besides, it is perhaps worth noting that typically $Q\gg\max_{u=1,...,U} T^u$, to keep the overhead of cyclic prefix relatively low.

We make the ideal assumption that all the users are perfectly synchronized. Then, the received OFDM symbol at the $n$-th BS antenna, with user $u$ considered, is
\begin{equation}
\mathbf{y}_n^u=\mathbf{C}_n^u\mathbf{x}^u+\sum_{v\neq u}\mathbf{C}_n^v\mathbf{x}^v+\mathbf{z}_n,
\end{equation}
where $\sum_{v\neq u}\mathbf{C}_n^v\mathbf{x}^v+\mathbf{z}_n
\sim\mathcal{CN}(\mathbf{0},\sigma^2\mathbf{I}_Q+\mathcal{E}_{\mathrm{s}}\sum_{v\neq u}\mathbf{C}_n^v(\mathbf{C}_n^{v})^{\dag})$ summarizes the co-channel interference and Gaussian noise experienced by user $u$. Noticing that $\mathbf{y}_n^u$ is actually identical for any $u\in\mathbb{U}$, here the superscript $u$ is simply for distinction with the single-user scenario. The SNR is defined as $\mathrm{SNR}\triangleq \mathcal{E}_{\mathrm{s}}/\sigma^2$, and without loss of generality, hereafter we let $\sigma^2=1$. Then, $\mathbf{y}_n^u$ is quantized by either a pair of high-resolution ADCs or a pair of one-bit ADCs, and the quantized output is
\begin{equation}
\mathbf{r}_n^u=\delta_n\mathbf{y}_n^u+\bar{\delta}_n\mathrm{sgn}(\mathbf{y}_n^u).
\end{equation}
We assume that there are still $K$ pairs of high-resolution ADCs available at the BS.

The quantized output $\mathbf{r}_n^u$ is then transformed to the frequency domain and later processed by user-specific frequency-domain equalization, leading to
\begin{equation}
\hat{\tilde{\mathbf{x}}}^u=\sum_{n=1}^N\mathbf{W}_n^u\mathbf{F}\mathbf{r}_n^u.
\end{equation}
At the decoder, a generalized nearest-neighbour decoding rule is adopted to decode the $u$-th user's transmit signal. That is, upon observing $\{\hat{\tilde{\mathbf{x}}}^u[l]\}_{l=1}^{L^u}$, the decoder computes, for each possible input message $m^u$ of user $u$, the following distance metric
\begin{equation}
D(m^u)=\frac{1}{L^u}\sum_{l=1}^{L^u}\|\hat{\tilde{\mathbf{x}}}^u[l]-a^u\tilde{\mathbf{x}}^u[m^u,l]\|^2,\ \ m^u\in\mathcal{M}^u,
\label{equ:equ_20}
\end{equation}
and determines $\hat{m}^u$ as the one that minimizes $D(m^u)$. Here $\mathcal{M}^u$ denotes the set of all the possible messages for user $u$, $\{\tilde{\mathbf{x}}^u[m^u,l]\}_{l=1}^{L^u}$ is the corresponding codeword for message $m^u$, and $L^u$ is the codeword length of user $u$. Note that $L^u$ in \eqref{equ:equ_20} can be different for different users.
\subsection{GMI and Optimal Equalizer}
In the multi-user scenario, if we take the co-channel interference (before quantization) as an additional colored Gaussian noise, then there is no essential difference with the single-user scenario, and the resulting equalizer can effectively handle this noise term by exploiting its statistical characteristics. The subsequent proposition summarizes our main results.
\begin{prop}
\label{prop:prop_4}
For given $\bm{\delta}$ and channel realizations $\{\mathbf{h}_n^v\}_{n=1}^N$, $v\in\mathbb{U}$, the GMI of user $u$ is given by
\begin{equation}
I_{\mathrm{GMI}}^u=\log\left(1+\frac{\Delta(\mathbf{w}_{\mathrm{opt}}^u,\bm{\delta})}{1-\Delta(\mathbf{w}_{\mathrm{opt}}^u,\bm{\delta})}\right).
\label{equ:prop_2.1}
\end{equation}
The performance indicator $\Delta(\mathbf{w}_{\mathrm{opt}}^u,\bm{\delta})$ and the optimal scaling parameter $a_{\mathrm{opt}}^u(\mathbf{w}_{\mathrm{opt}}^u,\bm{\delta})$ are
\begin{equation}
\Delta(\mathbf{w}_{\mathrm{opt}}^u,\bm{\delta})=a_{\mathrm{opt}}^u(\mathbf{w}_{\mathrm{opt}}^u,\bm{\delta})
=\frac{1}{Q\mathcal{E}_{\mathrm{s}}}(\mathbf{g}^{u})^{\dag}(\mathbf{D}^{u})^{-1}\mathbf{g}^u
\label{equ:prop_2.2}
\end{equation}
and are achieved by the optimal linear frequency-domain equalizer
\begin{equation}
\mathbf{w}_{\mathrm{opt}}^u=(\mathbf{D}^{u})^{-1}\mathbf{g}^u.
\label{equ:prop_2.3}
\end{equation}
In the above, $\mathbf{g}^u\triangleq[(\mathbf{g}_1^{u})^t,...,(\mathbf{g}_N^{u})^t]^t\in\mathbb{C}^{NQ\times 1}$, of which the $n$-th segment $\mathbf{g}_n^u$ is given by
\begin{equation}
\mathbf{g}_n^u=\delta_n\cdot\mathcal{E}_{\mathrm{s}}(\bm{\lambda}_n^{u})^*+\bar{\delta}_n\cdot\sqrt{\frac{2}{\pi}}\frac{\mathcal{E}_{\mathrm{s}}(\bm{\lambda}_n^{u})^*}{\sqrt{1+\sum_{v=1}^U\mathcal{E}_{\mathrm{s}}\|\bm{\lambda}_n^v\|^2/Q}}.
\label{equ:prop_2.4}
\end{equation}
The block matrix $\mathbf{D}^u\in\mathbb{C}^{NQ\times NQ}$ is identical for all the users:
\begin{equation}
\mathbf{D}^u\triangleq
\begin{pmatrix}
    \mathbf{D}_{11}^u & \mathbf{D}_{21}^u & \cdots & \mathbf{D}_{N1}^u \\
    \mathbf{D}_{12}^u & \mathbf{D}_{22}^u & \cdots & \mathbf{D}_{N2}^u \\
    \vdots  & \vdots  & \ddots & \vdots  \\
    \mathbf{D}_{1N}^u & \mathbf{D}_{2N}^u & \cdots & \mathbf{D}_{NN}^u
\end{pmatrix},
\label{equ:prop_2.5}
\end{equation}
and each of its blocks is a $Q$-dimensional diagonal matrix as
\begin{equation}
(\mathbf{D}_{nm}^u)_{qq}=(\mathbf{F}\mathbf{R}_{nm}^u\mathbf{F}^{\dag})_{qq}.
\label{equ:prop_2.6}
\end{equation}
Here $\mathbf{R}_{nm}^u\triangleq\mathbb{E}[\mathbf{r}_n^u(\mathbf{r}_m^{u})^\dag]$ is the correlation matrix of $\mathbf{r}_n^u$ and $\mathbf{r}_m^u$, and is also identical for all the users.
\end{prop}

\begin{proof}
As aforementioned, there is no fundamental difference between the single-user and the multi-user scenarios. Therefore, we only outline the sketch of the proof. For user $u$, the GMI also follows from \eqref{equ:equ_16}, except that we need to replace \eqref{equ:equ_17} by
\begin{equation}
\Delta(\mathbf{w}^u,\bm{\delta})=\frac{|\mathbb{E}[(\hat{\tilde{\mathbf{x}}}^{u})^{\dag}\tilde{\mathbf{x}}^u]|^2}
{Q\mathcal{E}_s\mathbb{E}[(\hat{\tilde{\mathbf{x}}}^{u})^{\dag}\hat{\tilde{\mathbf{x}}}^u]}.
\end{equation}
The calculation of $\mathbb{E}[(\hat{\tilde{\mathbf{x}}}^{u})^{\dag}\tilde{\mathbf{x}}^u]$ and $\mathbb{E}[(\hat{\tilde{\mathbf{x}}}^{u})^{\dag}\hat{\tilde{\mathbf{x}}}^u]$ follows essentially the same line as that in the single-user scenario, except that we replace $\mathbf{Y}_{nm}$ in \eqref{equ:equ_9} with
\begin{equation}
\mathbf{Y}_{nm}^u=
\begin{cases}
\mathbf{I}_Q+\mathcal{E}_{\mathrm{s}}\sum_{v=1}^{U}\mathbf{C}_n^v(\mathbf{C}_n^{v})^{\dag}, & n=m\in\mathbb{N} \\
\mathcal{E}_{\mathrm{s}}\sum_{v=1}^{U}\mathbf{C}_n^v(\mathbf{C}_m^{v})^{\dag}, & n\neq m\in\mathbb{N}.
\end{cases}
\end{equation}
Then following the same technical route as Appendix-A, we formulate $\Delta(\mathbf{w}^u,\bm{\delta})$ as a generalized Rayleigh quotient of $\mathbf{w}^u$, and based on which obtain the optimal linear frequency-domain equalizer $\mathbf{w}_{\mathrm{opt}}^u$ as well as the corresponding $\Delta(\mathbf{w}_{\mathrm{opt}}^u,\bm{\delta})$.
\end{proof}
\subsection{ADC Switch Scheme}
In the multi-user scenario, there is no generally convincing ADC switch scheme. Moreover, the high-dimensional property of the problem calls for efficient ADC switch scheme, and time-consuming schemes such as exhaustive search or greedy algorithm are practically infeasible. Therefore, we consider the following heuristic schemes.
\begin{itemize}
  \item Random ADC switch: high-resolution ADCs are switched to $K$ randomly chosen antennas.
  \item Norm-based ADC switch: switch the high-resolution ADCs to the antennas with the maximum $\sum_{v=1}^U\|\bm{\lambda}_n^v\|^2$ or $\sum_{v=1}^U\|\mathbf{h}_n^v\|^2$, equivalently.
\end{itemize}

The norm-based ADC switch scheme is asymptotically optimal in terms of sum rate in the low SNR regime, and achieves slightly better performance than the random ADC switch scheme. As the SNR or the number of users increases, the performance gain tends to decrease. Numerical results will be presented in Section \ref{sect:numerical} to examine the performance of the norm-based ADC switch scheme. Meanwhile, numerical study for the random ADC switch scheme is omitted due to space limitation.
\subsection{GMI Lower and Upper Bounds for Ergodic Time-Varying Channels}
Analogously we extend the analytical results to ergodic time-varying channels. Round-robin channel training is performed across the users and the BS antennas. Similarly we assume i.i.d. Rayleigh fading, i.e., $h_{nt}^u\sim\mathcal{CN}(0,1/T)$, for any $n\in\mathbb{N}$, $t\in\mathbb{T}$, and $u\in\mathbb{U}$. The BS also adopts an MMSE estimator, and we decompose $\mathbf{h}_n^u$ into
\begin{equation}
\mathbf{h}_n^u=\hat{\mathbf{h}}_n^u+\tilde{\mathbf{h}}_n^u,
\end{equation}
where $\hat{\mathbf{h}}_n^u$ is the estimated channel vector and $\tilde{\mathbf{h}}_n^u$ is the independent error vector. Under the Rayleigh fading assumption, it is easy to verify that both $\hat{\mathbf{h}}_n^u$ and $\tilde{\mathbf{h}}_n^u$ are complex Gaussian distributed. Moreover, if we define the normalized MSE as $\mathrm{MSE}_h=\sigma_h^2$, then $\hat{h}_{nt}^u\sim\mathcal{CN}(0,(1-\sigma_h^2)/T)$ and $\tilde{h}_{nt}^u\sim\mathcal{CN}(0,\sigma_h^2/T)$, for any $n\in\mathbb{N}$, $t\in\mathbb{T}$, and $u\in\mathbb{U}$.

The lower and upper bounds of the GMI are given by the following proposition, and numerical study will be conducted in Section \ref{sect:numerical} to verify their tightness.
\begin{prop}
For ergodic time-varying channels with estimated CSI, lower and upper bounds of the GMI of user $u$ are given by
\begin{eqnarray}
I_{\mathrm{GMI}}^{u,\mathrm{l}}&\!\!\!=\!\!\!&\rho\log\left(1\!+\!\frac{\mathbb{E}_{\hat{\mathbf{h}}}[\Delta(\mathbf{w}_{\mathrm{opt}}^{u,\mathrm{im}},\bm{\delta})]}{1\!-\!\mathbb{E}_{\hat{\mathbf{h}}}[\Delta(\mathbf{w}_{\mathrm{opt}}^{u,\mathrm{im}},\bm{\delta})]}\right),\\
I_{\mathrm{GMI}}^{u,\mathrm{u}}&\!\!\!=\!\!\!&\rho\mathbb{E}_{\hat{\mathbf{h}}}\left[\log\left(1\!+\!\frac{\Delta(\mathbf{w}_{\mathrm{opt}}^{u,\mathrm{im}},\bm{\delta})}{1\!-\!\Delta(\mathbf{w}_{\mathrm{opt}}^{u,\mathrm{im}},\bm{\delta})}\right)\right].
\end{eqnarray}
In the above, $\rho$ accounts for the overhead of channel training. Moreover, $\Delta(\mathbf{w}_{\mathrm{opt}}^{u,\mathrm{im}},\bm{\delta})$ and $\mathbf{w}_{\mathrm{opt}}^{u,\mathrm{im}}$ also come from \eqref{equ:prop_2.2} and \eqref{equ:prop_2.3}, but we need to replace $\mathbf{g}^u$ with $\mathbb{E}_{\tilde{\mathbf{h}}}[\mathbf{g}^u|\hat{\mathbf{h}}]$, and replace $\mathbf{D}$ with $\mathbb{E}_{{\tilde{\mathbf{h}}}}[\mathbf{D}|\hat{\mathbf{h}}]$. Here, $\hat{\mathbf{h}}$ and $\tilde{\mathbf{h}}$ denote $\hat{\mathbf{h}}_n^v$, $n\in\mathbb{N}$, $v\in\mathbb{U}$, and $\tilde{\mathbf{h}}_n^v$, $n\in\mathbb{N}$, $v\in\mathbb{U}$ respectively.
\end{prop}

In practice, the fast fading can be approximated to be piecewise-constant over $N_\mathrm{s}$ successive subcarriers since typically $Q\gg T$, to keep the overhead of cyclic prefix relatively low. Within each piecewise-constant frequency interval, only one pilot symbol is required for each user. Therefore we can reduce the channel training overhead by letting different users' pilot symbols occupy non-overlapped subcarriers. For more details, please refer to \cite[Sec. II-C]{marzetta2010noncooperative}. In this manner, the channel training overhead turns out to be $\lceil N/K\rceil\cdot\lceil U/N_{\mathrm{s}}\rceil$ and accordingly $\rho=\frac{T_{\mathrm{c}}-\lceil{N/K}\rceil\cdot\lceil{U/N_{\mathrm{s}}}\rceil}{T_{\mathrm{c}}}$.\footnote{For a typical system configuration, e.g., $N=64$, $K=16$, $U=10$, and $N_{\mathrm{s}}=14$, the proposed channel training method consumes four OFDM symbol intervals. As a comparison, the channel training in the ideal conventional architecture only needs one OFDM symbol interval and the channel training in \cite{studer2016quantized} lasts for $U=10$ OFDM symbol intervals.}
\section{Discussion on Computational Complexity}
\label{sect:complexity}
Since the computational complexity of matrix inversion grows cubically with the dimension of the matrix, directly inverting the $NQ\times NQ$ matrix $\mathbf{D}$ may impose great computational burden on the receiver. In the following, we exploit the specific structure of $\mathbf{D}$, and propose a reduced-complexity algorithm to efficiently compute $\mathbf{D}^{-1}$. With the proposed algorithm, the computational complexity of $\mathbf{D}^{-1}$ can be reduced from $O(N^3Q^3)$ to $O(N^3Q)$, significantly alleviating the computational burden of the receiver.
\begin{figure*}
  \centering
  \includegraphics[width=0.9 \textwidth]{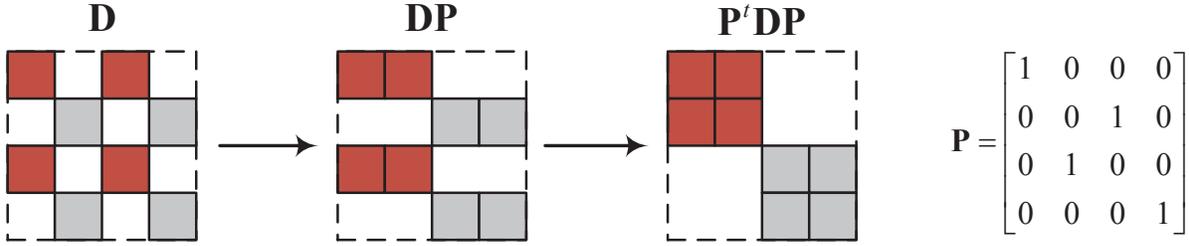}\\
  \caption{Transform the matrix $\mathbf{D}$ into a block diagonal matrix by row and column permutations: an example for $N=Q=2$.}
  \label{fig:permutation}
\end{figure*}

\begin{prop}
\label{prop:complexity}
The block matrix $\mathbf{D}$ consists of $N\times N$ blocks, of which each block is a $Q$-dimensional diagonal matrix. Exploiting this structure, the evaluation of $\mathbf{D}^{-1}$ can be simplified by applying $\mathbf{D}^{-1}=\mathbf{P}(\mathbf{P}^t\mathbf{DP})^{-1}\mathbf{P}^t$, where $\mathbf{P}$ is a specific permutation matrix that makes $\mathbf{P}^t\mathbf{DP}$ a block diagonal matrix.\footnote{Noticing that the permutation matrix $\mathbf{P}$ has exactly one element of 1 in each row and each column and 0s elsewhere, the computational complexity of $\mathbf{P}^t\mathbf{DP}$ is virtually negligible.}
\end{prop}
\begin{proof}
Observing the structure of $\mathbf{D}$, we notice that $\mathbf{D}$ can be transformed into a block diagonal matrix by some row-permutation matrix $\mathbf{P}_{\mathrm{r}}$ and some column-permutation matrix $\mathbf{P}_{\mathrm{c}}$. Moreover, it is straightforward that $\mathbf{P}_{\mathrm{r}}^t=\mathbf{P}_{\mathrm{c}}\triangleq\mathbf{P}$, due to the symmetry of $\mathbf{D}$. Figure \ref{fig:permutation} gives an example for $N=Q=2$, where the zero elements of $\mathbf{D}$ are left blank and the nonzero entries belonging to the same subcarrier are marked by the same color.

Since the permutation matrix $\mathbf{P}$ solely depends on the system parameter $(N,Q)$, it can be saved offline thus incuring no additional computational burden on the receiver. Moreover, the permutation matrix $\mathbf{P}$ can be found by a simple training program. To this end, we create a training matrix $\mathbf{G}$ that shares the same structure as $[\mathbf{D}_{11},...,\mathbf{D}_{N1}]$; that is, $\mathbf{G}\triangleq [\mathbf{G}_1,...,\mathbf{G}_N]$, and $\mathbf{G}_n$ is a $Q$-dimensional diagonal matrix for any $n\in\mathbb{N}$. Particularly, the value of each nonzero element of $\mathbf{G}$ indicates its column index after performing the column permutation $\mathbf{P}$. Exploiting these inherent marks we are able to transform $\mathbf{G}$ into $\mathbf{GP}$ by no more than $NQ$ times of pair-wise column permutations and obtain $\mathbf{P}$ during the training process.
\end{proof}

Before inverting $\mathbf{D}$, evaluating $\mathbf{D}$ itself also incurs a computational complexity of $O(N^2Q^2\log_2 Q)$. Beyond that, the computational burden of $\mathbf{g}$ and $\mathbf{D}^{-1}\mathbf{g}$ is negligible. Therefore, with Proposition \ref{prop:complexity}, the computational complexity of $\mathbf{w}_{\mathrm{opt}}$ can be reduced from $O(N^3 Q^3)$ to $O(\max\{N^3 Q, N^2 Q^2\log_2 Q\})$, significantly alleviating the computational burden of the receiver.

\section{Numerical Study}
\label{sect:numerical}
Now we corroborate the analytical results through extensive numerical studies. All the figures in this section are for ergodic time-varying channels, and the channel coefficients are drawn from i.i.d. Rayleigh fading, i.e., $h_{nt}^v\sim\mathcal{CN}(0,1/T)$, for any $n\in\mathbb{N}$, $t\in\mathbb{T}$, and $v\in\mathbb{U}$. The norm-based ADC switch scheme is adopted in both single-user and multi-user scenarios. ``CA'' and ``AS'' represent ideal conventional architecture and antenna selection, respectively.
\begin{figure}
  \centering
  \includegraphics[width=0.45 \textwidth]{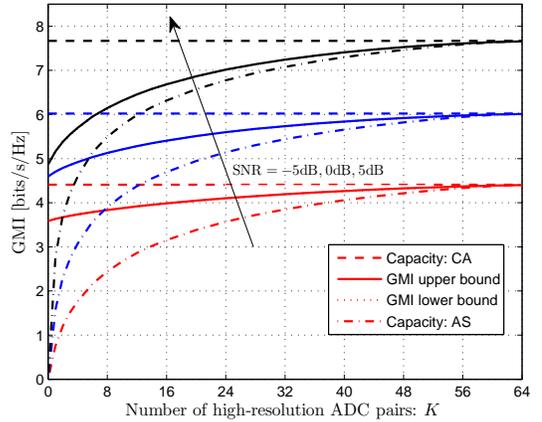}\\
  \caption{GMI for different numbers of high-resolution ADC pairs: perfect CSI, $N=64$, $U=1$, $Q=32$, and $T=5$.}
  \label{fig:fig_2}
\end{figure}
\begin{figure}
  \centering
  \includegraphics[width=0.45 \textwidth]{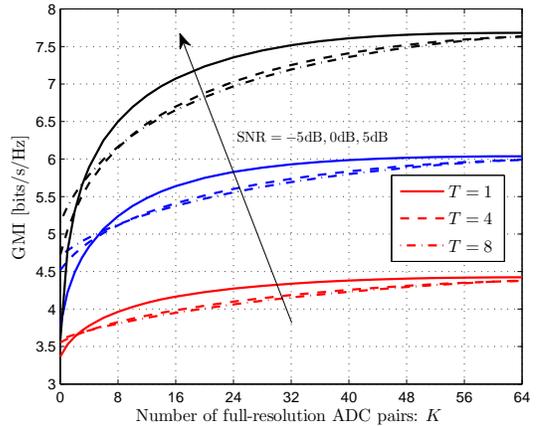}\\
  \caption{GMI lower bound for different numbers of high-resolution ADC pairs: perfect CSI, $N=64$, $U=1$, and $Q=32$. Three different choices of $T$ are made, of which $T=1$ means frequency-flat fading, $T=8$ corresponds to frequency-selective fading in a rich-scattering environment, and $T=4$ accounts for a mediate scenario.}
  \label{fig:fig_3}
\end{figure}
\subsection{Single-User Scenario}
Assuming perfect CSI at the BS, Figure \ref{fig:fig_2} displays the GMI of the mixed-ADC architecture for different numbers of high-resolution ADC pairs. Several observations are in order. First we notice that the GMI lower and upper bounds are very tight, and thus we will only use the GMI lower bound in the subsequent evaluation. In addition, for the special case of $K=N$, there is a barely visible gap between the GMI and the capacity, as predicted by Corollary \ref{cor:cor_1}. Moreover, the mixed-ADC architecture with a small proportion of high-resolution ADCs does achieve a dominant portion of the capacity of ideal conventional architecture, and significantly outperforms antenna selection with the same number of high-resolution ADCs.
\begin{figure}
  \centering
  \includegraphics[width=0.45 \textwidth]{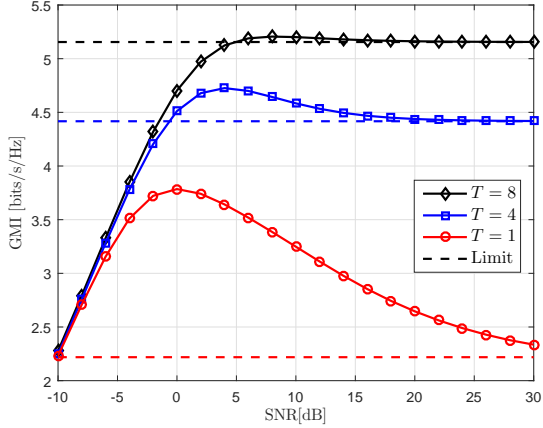}\\
  \caption{GMI lower bound for various SNRs: perfect CSI, $N=64$, $K=0$, $U=1$, $Q=32$, and $T=1$, $4$, $8$.}
  \label{fig:fig_4}
\end{figure}
\begin{figure}
  \centering
  \includegraphics[width=0.45 \textwidth]{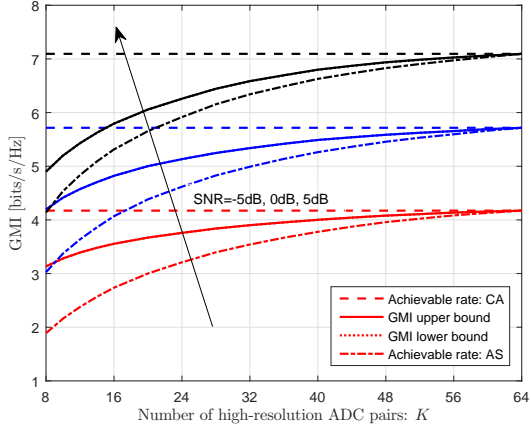}\\
  \caption{GMI for different numbers of high-resolution ADC pairs: imperfect CSI, $N=64$, $U=1$, $Q=32$, $T=5$, $\mathrm{MSE}_h=-10$dB, and $T_{\mathrm{c}}=53$.}
  \label{fig:fig_5}
\end{figure}
\begin{figure}
  \centering
  \includegraphics[width=0.45 \textwidth]{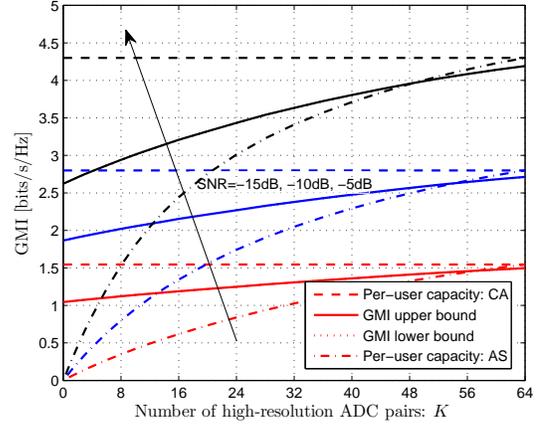}\\
  \caption{Per-user GMI for different numbers of high-resolution ADC pairs: perfect CSI, $N=64$, $U=10$, $Q=32$, and $T^v=5$ for any $v\in\mathbb{U}$.}
  \label{fig:fig_7}
\end{figure}
\begin{figure}
  \centering
  \includegraphics[width=0.45 \textwidth]{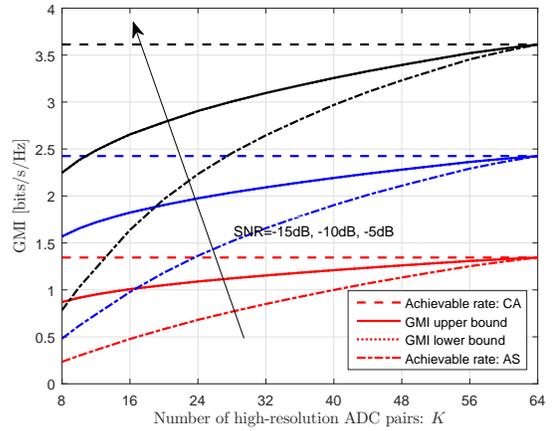}\\
  \caption{Per-user GMI for different numbers of high-resolution ADC pairs: imperfect CSI, $N=64$, $U=10$, $Q=32$, $\mathrm{MSE}_h=-10$dB, $T_{\mathrm{c}}=53$, and $T^v=5$ for any $v\in\mathbb{U}$.}
  \label{fig:fig_8}
\end{figure}

The impact of frequency diversity on the system performance is addressed by Figure \ref{fig:fig_3}, where three different choices of $T$ are made. For each given SNR, we notice that there is an intersection between two of the curves. Particularly, if $K$ lies at the right side of the intersection, a larger $T$ would lead to a lower GMI. This may be attributed to the limitation of the linear frequency-domain equalizer in mitigating ICI. If $K$ lies at the left side of the intersection, on the other hand, a larger $T$ would achieve a higher GMI. Because in this situation, there are few high-resolution ADCs and, hence, frequency diversity becomes crucial for signal recovery at the receiver.

By letting $K=0$ and varying the SNR, Figure \ref{fig:fig_4} gives a closer look at the impact of frequency diversity. The dashed lines correspond to the limits of the GMI in the high SNR regime, as given by Corollary \ref{cor:cor_4}. First we notice that, for each given $T$, the GMI will increase first and then turn down as the SNR grows large. Such a phenomenon has been observed in frequency-flat one-bit massive SIMO systems, e.g., \cite{jacobsson2015one}-\cite{jacobsson2016massive} and \cite{liang2015mixed}. As explained in the aforementioned works, for one-bit massive SIMO, the amplitude information of the transmit signal tends to be totally lost as the SNR approaches infinity; see also \cite{mo2014high}. That is, in this situation a moderate amount of noise is actually beneficial for signal recovery at the receiver. Further we notice that, over the entire SNR regime, a larger $T$ will always achieve a higher GMI. Because for one-bit massive SIMO, frequency diversity is a key factor that enables signal recovery at the receiver.

Figure \ref{fig:fig_5} examines the impact of imperfect CSI on the system performance, assuming $\mathrm{MSE}_h=-10$dB and $T_{\mathrm{c}}=$53.\footnote{We adopt Jake's model and assume $f_c=2$GHz, with the OFDM symbol interval being 71.4$\mu$s and the users' speed being $60$km/h. As a comparison, the authors of \cite{jacobsson2016massive} assume a coherence interval as long as 1142 symbols.} We observe that the performance gap between the mixed-ADC architecture and the ideal conventional architecture is slightly enlarged, mainly due to the increase of the channel training overhead. Nevertheless, the conclusions we made for the perfect CSI case still hold here.
\subsection{Multi-user Scenario}
Several observations are in order from Figure \ref{fig:fig_7}. First, we notice that the GMI lower and upper bounds again virtually coincide with each other. Second, due to not applying successive interference cancellation (SIC) at the receiver, there is a visible but marginal performance loss at $K=N$ when compared with the per-user capacity\footnote{The per-user capacity is obtained through dividing the sum-capacity by the number of users.} of the ideal conventional architecture. Finally, the mixed-ADC architecture with a small proportion of high-resolution ADCs still achieves a large portion of the per-user capacity of the ideal conventional architecture, while overwhelms antenna selection with the same number of high-resolution ADCs. Such observations also hold for the imperfect CSI case, as verified by Figure \ref{fig:fig_8}.
\begin{figure}
  \centering
  \includegraphics[width=0.45 \textwidth]{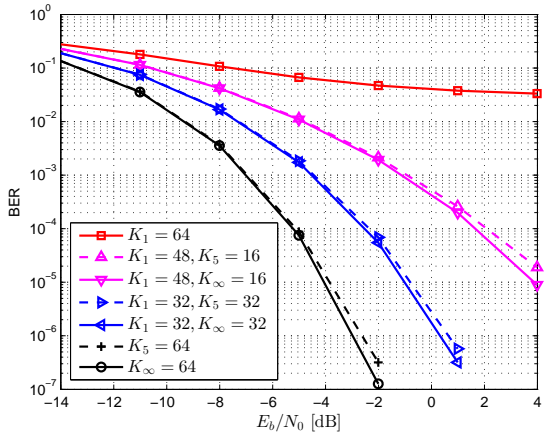}\\
  \caption{BER performance in the multi-user scenario: perfect CSI, $N=64$, $U=10$, $Q=32$, and $T^v=5$ for any $v\in\mathbb{U}$.}
  \label{fig:fig_9}
\end{figure}
\begin{figure}
  \centering
  \includegraphics[width=0.45 \textwidth]{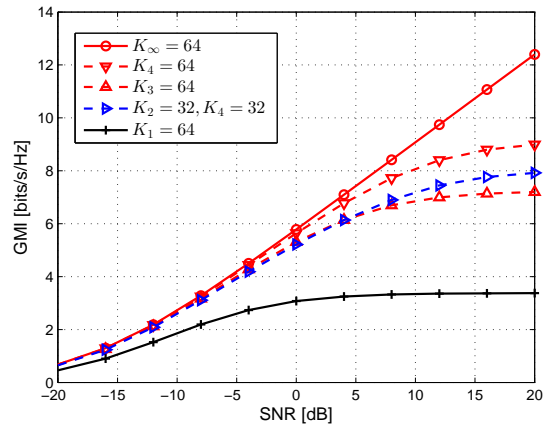}\\
  \caption{Per-user GMI under different ADC configurations: perfect CSI, $N=64$, $U=10$, $Q=32$, and $T^v=5$ for any $v\in\mathbb{U}$.}
  \label{fig:fig_10}
\end{figure}

We further examine the BER performance of the mixed-ADC architecture in the multi-user scenario. To this end, we assume that each user adopts an independent $(2, 1, 3)$ convolutional coder, where the code rate is $1/2$, the constraint length is $3$, and the generator polynomials are $(1, 1, 0)$ and $(1, 1, 1)$. 16-QAM modulation with Gray mapping is adopted to map the coded bits into system input $\tilde{\mathbf{x}}$.\footnote{Noting that the optimal linear frequency-domain equalizer in Proposition \ref{prop:prop_4} is with respect to Gaussian distributed channel inputs, it is mismatched with 16-QAM modulation and is therefore suboptimal in this situation. Nevertheless, benefiting from the central-limit theorem, the actual channel inputs, $\mathbf{F}^{\dag}\tilde{\mathbf{x}}$, may be approximately viewed as Gaussian as $Q$ grows large. As a result, the performance loss due to this mismatch is expected to be marginal.} In this manner, two information bits are first encoded into four codeword bits and then mapped into a 16-QAM symbol. As a result, $E_{\mathrm{b}}/N_0$ under this setup equals $\mathrm{SNR}-$3dB. Hard-decision Viterbi decoding is performed at the BS, also in a per-user manner.

Numerical result is presented in Figure \ref{fig:fig_9}, where $K_i$ is the number of $i$-bit ADC pairs. The quantization bins and output levels of $i$-bit ADC are given by \cite[Tab.-I]{max1960quantizing}. We notice that one-bit massive MIMO suffers from error floor, as already revealed in \cite{zhang2015mixed}-\cite{wang2014multiuser}. The mixed-ADC architecture, on the other hand, remarkably improves the BER performance. Performance loss due to replacing high-resolution ADCs by 5-bit ADCs is also examined, still using the equalizer derived in Proposition \ref{prop:prop_4}. Such a mismatched equalizer entails relatively low computational complexity, and incurs marginal BER loss as verified by Figure \ref{fig:fig_9}.\footnote{Similarly, the authors of \cite{studer2016quantized} observed that the standard OFDM processing (i.e., ignoring the quantizer) actually achieves good BER performance for 4-bit quantization and beyond.} These observations again validate the merits of the mixed-ADC architecture.
\subsection{GMI under Different ADC Configurations}
We note that, beyond the mixed-ADC architecture specialized in this paper, the GMI analytical framework established is also applicable to any other ADC configuration. For any other kind of ADC configuration, calculation of the GMI still follows from the general idea of Proposition \ref{prop:prop_4}, and only $\mathbf{g}_n^u$ and $\mathbf{R}_{nm}^u$ will change along with the ADC configuration. Moreover, $\mathbf{g}_n^u$ always has a closed-form expression. $\mathbf{R}_{nm}^u$ has a closed-form expression if the BS adopts one-bit or high-resolution ADCs, otherwise we have to rely on numerical integrations to accurately evaluate $\mathbf{R}_{nm}^u$.

Figure \ref{fig:fig_10} displays the per-user GMI under different ADC configurations, assuming perfect CSI at the BS. Unlike Figure \ref{fig:fig_9}, here each equalizer is matched with the corresponding ADC configuration. It is clear that one-bit massive MIMO generally has to tolerate large rate losses for target spectral efficiency (TSE) above 2 bits/s/Hz. Four-bit massive MIMO, on the other hand, only incurs marginal rate losses for TSE below 7 bits/s/Hz.

Comparison between the homogeneous-ADC architecture and the mixed-ADC architecture is also conducted, taking $K_3=64$ and $K_2=32$, $K_4=32$ as an example. Note that hardware costs of these two configurations are close. Figure \ref{fig:fig_10} reveals that these two configurations achieve nearly the same performance for TSE below 6 bits/s/Hz and the mixed-ADC architecture performs better for TSE above 6 bits/s/Hz. A comprehensive comparison between the homogeneous-ADC architecture and the mixed-ADC architecture is left for future work due to space limitation.
\section{Conclusion}
\label{sect:conclusion}
In this paper, we developed an analytical framework for the mixed-ADC architecture operating over frequency-selective channels. Notably, the analytical framework is also applicable to any other kind of ADC configuration. Extensive numerical studies demonstrate that the mixed-ADC architecture is able to achieve performance close to the ideal conventional architecture, and thus we envision it as a promising option for effective design of massive MIMO receivers.

Beyond the scope of this paper, several important problems need further investigation. First, for a given TSE, optimization of the bit-width and ratio of each kind of ADC adopted will further reduce the hardware cost and energy consumption of the mixed-ADC architecture. Second, more efficient channel estimation algorithm and more effective ADC switch scheme will further improve the performance of the mixed-ADC architecture, especially in the multi-user scenario. Third, another line of work advocates the homogeneous-ADC architecture for energy-efficient design of massive MIMO, and therefore a reasonable and comprehensive comparison between the mixed-ADC architecture and the homogeneous-ADC architecture is particularly important, especially when taking practical issues such as time/frequency synchronization and channel estimation into account.

\section*{Appendix}
\label{sect:appendix}
\subsection{Proof of Proposition 2}
\label{subsect:proposition_1}
We first introduce a lemma, with which we are able to derive a closed-form expression of $\Delta(\mathbf{w},\bm{\delta})$. The proof is similar as those for \cite[Lem. 1\&2]{liang2015mixed}, and thus is omitted due to space limitation.
\begin{lem}
\label{lem:lem_1}
For bivariate circularly-symmetric complex Gaussian vector
\begin{equation}
\begin{pmatrix}
u_1\\
u_2
\end{pmatrix}
\sim\mathcal{CN}
\left(\mathbf{0},
\begin{pmatrix}
\sigma_1^2 & \sigma_{12} \\
\sigma_{12}^{\dag} & \sigma_2^2
\end{pmatrix}
\right),
\end{equation}
we have
\begin{equation}
\mathbb{E}[\mathrm{sgn}^{\dag}(u_1)u_2]=\sqrt{\frac{2}{\pi}}\frac{\sigma_{12}^{\dag}}{\sigma_1},
\label{equ:lemma_1}
\end{equation}
and
\begin{equation}
\mathbb{E}[\mathrm{sgn}(u_1)\mathrm{sgn}^{\dag}(u_2)]=\frac{2}{\pi}[\arcsin(\theta_{\mathrm{R}})+j\arcsin(\theta_{\mathrm{I}})],
\label{equ:lemma_2}
\end{equation}
where $\theta_{\mathrm{R}}$ and $\theta_{\mathrm{I}}$ are the real and imaginary parts of the correlation coefficient $\theta=\frac{\sigma_{12}}{\sigma_1 \sigma_2}$ respectively, and $\sigma_{12}$ is defined as $\sigma_{12}\triangleq\mathbb{E}[u_1u_2^{\dag}]$.
\end{lem}

Hereafter we denote by $\mathbb{N}_{\mathrm{1}}$ the set of indexes that make $\delta_n=1$ while by $\mathbb{N}_{\mathrm{0}}$ the set of indexes that make $\delta_n=0$. For the numerator $\mathbb{E}[\hat{\tilde{\mathbf{x}}}^{\dag}\tilde{\mathbf{x}}]$, with some manipulation we have
\begin{eqnarray}
&&\mathbb{E}[\hat{\tilde{\mathbf{x}}}^{\dag}\tilde{\mathbf{x}}]\nonumber\\
&\!\!\!=\!\!\!&\sum_{n\in\mathbb{N}}\mathbb{E}[\mathbf{r}_n^{\dag}\mathbf{F}^{\dag}\mathbf{W}_n^{\dag}\tilde{\mathbf{x}}]\nonumber\\
&\!\!\!=\!\!\!&\sum_{n\in\mathbb{N}_{\mathrm{1}}}\!\!\mathbb{E}[\mathbf{y}_n^{\dag}\mathbf{F}^{\dag}\mathbf{W}_n^{\dag}\tilde{\mathbf{x}}]
\!+\!\sum_{n\in\mathbb{N}_{\mathrm{0}}}\!\!\mathbb{E}[\mathrm{sgn}^{\dag}(\mathbf{y}_n)\mathbf{F}^{\dag}\mathbf{W}_n^{\dag}\mathbf{F}\mathbf{x}],
\label{equ:equ_1}
\end{eqnarray}
where the first term is contributed by the antennas connected with high-resolution ADCs, and the second term comes from the antennas connected with one-bit ADCs. In the following, we need to evaluate them separately.

First let us look at $\mathbb{E}[\mathbf{y}_n^{\dag}\mathbf{F}^{\dag}\mathbf{W}_n^{\dag}\tilde{\mathbf{x}}]$. With some manipulations we have
\begin{eqnarray}
\mathbb{E}[\mathbf{y}_n^{\dag}\mathbf{F}^{\dag}\mathbf{W}_n^{\dag}\tilde{\mathbf{x}}]
&=&\mathbb{E}[\mathbf{x}^{\dag}\mathbf{C}_n^{\dag}\mathbf{F}^{\dag}\mathbf{W}_n^{\dag}\tilde{\mathbf{x}}]
+\mathbb{E}[\mathbf{z}_n^{\dag}\mathbf{F}^{\dag}\mathbf{W}_n^{\dag}\tilde{\mathbf{x}}]\nonumber\\
&\overset{(a)}{=}&\mathbb{E}[\tilde{\mathbf{x}}^{\dag}\bm{\Lambda}_n^{\dag}\mathbf{W}_n^{\dag}\tilde{\mathbf{x}}]\nonumber\\
&=&\mathrm{tr}\left(\bm{\Lambda}_n^{\dag}\mathbf{W}_n^{\dag}\mathbb{E}[\tilde{\mathbf{x}}\tilde{\mathbf{x}}^{\dag}]\right)\nonumber\\
&=&\mathcal{E}_{\mathrm{s}}\mathrm{tr}\left(\bm{\Lambda}_n^{\dag}\mathbf{W}_n^{\dag}\right)
=\mathcal{E}_{\mathrm{s}}\mathbf{w}_n^{\dag}\bm{\lambda}_n^*,
\label{equ:equ_2}
\end{eqnarray}
where (a) follows from the independence of $\tilde{\mathbf{x}}$ and $\mathbf{z}_n$, the relationship $\mathbf{x}=\mathbf{F}^{\dag}\tilde{\mathbf{x}}$, and the decomposition $\mathbf{C}_n=\mathbf{F}^{\dag}\bm{\Lambda}_n\mathbf{F}$.

Next, we turn to $\mathbb{E}[\mathrm{sgn}^{\dag}(\mathbf{y}_n)\mathbf{F}^{\dag}\mathbf{W}_n^{\dag}\mathbf{F}\mathbf{x}]$. To start with, we define $\mathbf{y}_n\triangleq[y_{n1},...,y_{nQ}]^t$. Then it is obvious that
\begin{equation}
\mathbb{E}[\mathrm{sgn}^{\dag}(\mathbf{y}_n)\mathbf{F}^{\dag}\mathbf{W}_n^{\dag}\mathbf{F}\mathbf{x}]
=\sum_{q=1}^Q\mathbb{E}[\mathrm{sgn}^{\dag}(y_{nq})(\mathbf{F}^{\dag}\mathbf{W}_n^{\dag}\mathbf{F}\mathbf{x})_q].
\end{equation}
By noticing that $y_{nq}\sim\mathcal{CN}(0,1+\mathcal{E}_{\mathrm{s}}(\mathbf{C}_n\mathbf{C}_n^{\dag})_{qq})$,
and moreover that $(\mathbf{C}_n\mathbf{C}_n^{\dag})_{qq}=(\mathbf{F}^{\dag}\bm{\Lambda}_n\bm{\Lambda}_n^{\dag}\mathbf{F})_{qq}=\|\bm{\lambda}_n\|^2/Q$, we obtain the distribution of $y_{nq}$ as
\begin{equation}
y_{nq}\sim\mathcal{CN}\left(0,1+\mathcal{E}_{\mathrm{s}}\|\bm{\lambda}_n\|^2/Q\right).
\end{equation}
Further, $y_{nq}$ and $(\mathbf{F}^{\dag}\mathbf{W}_n^{\dag}\mathbf{F}\mathbf{x})_q$ are jointly circularly symmetric complex Guassian, with their covariance being
\begin{eqnarray}
\mathbb{E}[y_{nq}^{\dag}(\mathbf{F}^{\dag}\mathbf{W}_n^{\dag}\mathbf{F}\mathbf{x})_q]
&=&(\mathbb{E}[\mathbf{F}^{\dag}\mathbf{W}_n^{\dag}\mathbf{F}\mathbf{x}\mathbf{y}_n^{\dag}])_{qq}\nonumber\\
&=&\mathcal{E}_{\mathrm{s}}(\mathbf{F}^{\dag}\mathbf{W}_n^{\dag}\bm{\Lambda}_n^{\dag}\mathbf{F})_{qq}\nonumber\\
&=&\mathcal{E}_{\mathrm{s}}\mathbf{w}_n^{\dag}\bm{\lambda}_n^*/Q.
\end{eqnarray}
As a result, exploiting \eqref{equ:lemma_1} we arrive at
\begin{equation}
\mathbb{E}[\mathrm{sgn}^{\dag}(\mathbf{y}_n)\mathbf{F}^{\dag}\mathbf{W}_n^{\dag}\mathbf{F}\mathbf{x}]
=\sqrt{\frac{2}{\pi}}\frac{\mathcal{E}_{\mathrm{s}}\mathbf{w}_n^{\dag}\bm{\lambda}_n^*}{\sqrt{1+\mathcal{E}_{\mathrm{s}}\|\bm{\lambda}_n\|^2/Q}}.
\label{equ:equ_3}
\end{equation}

Now, we are allowed to combine \eqref{equ:equ_1}-\eqref{equ:equ_2} and \eqref{equ:equ_3} to obtain $\mathbb{E}[\hat{\tilde{\mathbf{x}}}^{\dag}\tilde{\mathbf{x}}]$, given as
\begin{equation}
\mathbb{E}[\hat{\tilde{\mathbf{x}}}^{\dag}\tilde{\mathbf{x}}]
=\sum_{n\in\mathbb{N}_{\mathrm{1}}}\mathcal{E}_{\mathrm{s}}\mathbf{w}_n^{\dag}\bm{\lambda}_n^*
+\sum_{n\in\mathbb{N}_{\mathrm{0}}}\sqrt{\frac{2}{\pi}}\frac{\mathcal{E}_{\mathrm{s}}\mathbf{w}_n^{\dag}\bm{\lambda}_n^*}{\sqrt{1+\mathcal{E}_{\mathrm{s}}\|\bm{\lambda}_n\|^2/Q}}.
\label{equ:equ_4}
\end{equation}
For the convenience of further investigation, we define $\mathbf{w}\triangleq[\mathbf{w}_1^t,...,\mathbf{w}_N^t]^t$ and rewrite \eqref{equ:equ_4} as
\begin{equation}
\mathbb{E}[\hat{\tilde{\mathbf{x}}}^{\dag}\tilde{\mathbf{x}}]=\mathbf{w}^{\dag}\mathbf{g},
\end{equation}
where $\mathbf{g}\triangleq[\mathbf{g}_1^t,...,\mathbf{g}_N^t]^t\in\mathbb{C}^{NQ\times 1}$ is given by
\begin{equation}
\mathbf{g}_n=\delta_n\cdot\mathcal{E}_{\mathrm{s}}\bm{\lambda}_n^*+\bar{\delta}_n\cdot\sqrt{\frac{2}{\pi}}\frac{\mathcal{E}_{\mathrm{s}}\bm{\lambda}_n^*}{\sqrt{1+\mathcal{E}_{\mathrm{s}}\|\bm{\lambda}_n\|^2/Q}}.
\label{equ:equ_15}
\end{equation}

In order to evaluate $\mathbb{E}[\hat{\tilde{\mathbf{x}}}^{\dag}\hat{\tilde{\mathbf{x}}}]$, we define the correlation matrix between $\mathbf{r}_n$ and $\mathbf{r}_m$ as $\mathbf{R}_{nm}\triangleq\mathbb{E}[\mathbf{r}_n\mathbf{r}_m^{\dag}]$. Then we have
\begin{eqnarray}
\mathbb{E}[\hat{\tilde{\mathbf{x}}}^{\dag}\hat{\tilde{\mathbf{x}}}]
&=&\sum_{m=1}^N\sum_{n=1}^N\mathbb{E}[\mathbf{r}_m^{\dag}\mathbf{F}^{\dag}\mathbf{W}_m^{\dag}\mathbf{W}_n\mathbf{F}\mathbf{r}_n]\nonumber\\
&=&\sum_{m=1}^N\sum_{n=1}^N\mathrm{tr}\left(\mathbf{W}_m^{\dag}\mathbf{W}_n\mathbf{F}\mathbb{E}[\mathbf{r}_n\mathbf{r}_m^{\dag}]\mathbf{F}^{\dag}\right)\nonumber\\
&=&\sum_{m=1}^N\sum_{n=1}^N\mathrm{tr}\left(\mathbf{W}_m^{\dag}\mathbf{W}_n\mathbf{F}\mathbf{R}_{nm}\mathbf{F}^{\dag}\right).
\end{eqnarray}
Noticing that $\mathbf{W}_n, n\in\mathbb{N}$ are all diagonal matrices, to get rid of the trace operation, we may define a diagonal matrix $\mathbf{D}_{nm}$ as
\begin{equation}
(\mathbf{D}_{nm})_{qq}=(\mathbf{F}\mathbf{R}_{nm}\mathbf{F}^{\dag})_{qq},
\label{equ:equ_8}
\end{equation}
and based on which rewrite $\mathbb{E}[\hat{\tilde{\mathbf{x}}}^{\dag}\hat{\tilde{\mathbf{x}}}]$ as
\begin{equation}
\mathbb{E}[\hat{\tilde{\mathbf{x}}}^{\dag}\hat{\tilde{\mathbf{x}}}]=\sum_{m=1}^N\sum_{n=1}^N\mathbf{w}_m^{\dag}\mathbf{D}_{nm}\mathbf{w}_n.
\label{equ:equ_5}
\end{equation}
This further motivates us to rewrite it in a compact manner
\begin{eqnarray}
&&\mathbb{E}[\hat{\tilde{\mathbf{x}}}^{\dag}\hat{\tilde{\mathbf{x}}}]\nonumber\\
&\!\!\!=\!\!\!&[\mathbf{w}_1^{\dag},\mathbf{w}_2^{\dag},...,\mathbf{w}_{N}^{\dag}]
\begin{pmatrix}
    \mathbf{D}_{11} & \mathbf{D}_{21} & \cdots & \mathbf{D}_{N1} \\
    \mathbf{D}_{12} & \mathbf{D}_{22} & \cdots & \mathbf{D}_{N2} \\
    \vdots  & \vdots  & \ddots & \vdots  \\
    \mathbf{D}_{1N} & \mathbf{D}_{2N} & \cdots & \mathbf{D}_{NN}
\end{pmatrix}
\begin{pmatrix}
    \mathbf{w}_1 \\
    \mathbf{w}_2 \\
    \vdots \\
    \mathbf{w}_N
\end{pmatrix}\nonumber\\
&\!\!\!\triangleq\!\!\!&\mathbf{w}^{\dag}\mathbf{D}\mathbf{w}.
\end{eqnarray}
We note that the block matrix $\mathbf{D}\in\mathbb{C}^{NQ\times NQ}$ is a Hermitian matrix and moreover each of its blocks is a $Q$-dimensional diagonal matrix.

Now, we are allowed to formulate $\Delta(\mathbf{w},\bm{\delta})$ as a generalized Rayleigh quotient of $\mathbf{w}$; that is,
\begin{equation}
\Delta(\mathbf{w},\bm{\delta})=\frac{\mathbf{w}^{\dag}\mathbf{g}\mathbf{g}^{\dag}\mathbf{w}}{Q\mathcal{E}_{\mathrm{s}}\mathbf{w}^{\dag}\mathbf{D}\mathbf{w}}.
\end{equation}
Then, exploiting a similar argument as that adopted in \cite[Prop. 3]{liang2015mixed}, we may easily obtain the optimal linear frequency-domain equalizer $\mathbf{w}_{\mathrm{opt}}$ and the corresponding $\Delta(\mathbf{w}_{\mathrm{opt}},\bm{\delta})$, as summarized by Proposition \ref{prop:prop_1}.

The evaluation of the matrix $\mathbf{D}$ remains unaccomplished. To this end, we first define $\mathbf{Y}_{nm}\triangleq\mathbb{E}[\mathbf{y}_n\mathbf{y}_m^{\dag}]$, and it is easy to verify that
\begin{equation}
\mathbf{Y}_{nm}=
\begin{cases}
\mathbf{I}_Q+\mathcal{E}_{\mathrm{s}}\mathbf{C}_n\mathbf{C}_n^{\dag},\ &\ n=m\in\mathbb{N},\\
\mathcal{E}_{\mathrm{s}}\mathbf{C}_n\mathbf{C}_m^{\dag},\ &\ n\neq m\in\mathbb{N}.
\end{cases}
\label{equ:equ_9}
\end{equation}
Then, we introduce a series of matrices $\bm{\Theta}_{nm}$, $n,m\in\mathbb{N}_{\mathrm{0}}$, of which $\bm{\Theta}_{nm}$ corresponds to the correlation coefficient matrix between $\mathbf{y}_n$ and $\mathbf{y}_m$, with its $(p,q)$-th element given by
\begin{eqnarray}
(\bm{\Theta}_{nm})_{pq}=\frac{(\mathbf{Y}_{nm})_{pq}}{\sqrt{(\mathbf{Y}_{nn})_{pp}}\sqrt{(\mathbf{Y}_{mm})_{qq}}}.
\label{equ:equ_6}
\end{eqnarray}
Due to the mixed nature of $\mathbf{r}$, the computation of $\mathbf{R}_{nm}$ for different $(n,m)$ may follow different routes, and therefore in the following, we need to evaluate them case by case.

Case 1: $n,m\in\mathbb{N}_{\mathrm{1}}$. In this case, $\mathbf{r}_n=\mathbf{y}_n$, $\mathbf{r}_m=\mathbf{y}_m$, and thus we have
\begin{equation}
\mathbf{R}_{nm}=\mathbf{Y}_{nm}.
\end{equation}

Case 2: $n,m\in\mathbb{N}_{\mathrm{0}}$. In this case, $\mathbf{r}_n=\mathrm{sgn}(\mathbf{y}_n)$ and $\mathbf{r}_m=\mathrm{sgn}(\mathbf{y}_m)$. Exploiting \eqref{equ:lemma_2} we get
\begin{eqnarray}
(\mathbf{R}_{nm})_{pq}&\!\!\!=\!\!\!&\mathbb{E}[\mathrm{sgn}(y_{np})\mathrm{sgn}^{\dag}(y_{mq})]\nonumber\\
&\!\!\!=\!\!\!&\frac{2}{\pi}[\arcsin((\bm{\Theta}_{nm})_{pq,\mathrm{R}})+j\arcsin((\bm{\Theta}_{nm})_{pq,\mathrm{I}})].\nonumber\\
\label{equ:equ_11}
\end{eqnarray}

Case 3: $n\in\mathbb{N}_{\mathrm{1}}$ and $m\in\mathbb{N}_{\mathrm{0}}$. In this case, $\mathbf{r}_n=\mathbf{y}_n$ and $\mathbf{r}_m=\mathrm{sgn}(\mathbf{y}_m)$. Applying \eqref{equ:lemma_1} we obtain
\begin{equation}
(\mathbf{R}_{nm})_{pq}=\mathbb{E}[\mathrm{sgn}^{\dag}(y_{mq})y_{np}]=(\mathbf{Y}_{nm})_{pq}\sqrt{\frac{2}{\pi(\mathbf{Y}_{mm})_{qq}}}.
\end{equation}

Case 4: $n\in\mathbb{N}_{\mathrm{0}}$ and $m\in\mathbb{N}_{\mathrm{1}}$. This case is similar to the former one, and with some manipulations we find that
\begin{equation}
(\mathbf{R}_{nm})_{pq}=\mathbb{E}[\mathrm{sgn}(y_{np})y_{mq}^{\dag}]=(\mathbf{Y}_{nm})_{pq}\sqrt{\frac{2}{\pi(\mathbf{Y}_{nn})_{pp}}}.
\end{equation}

In summary, we enumerate $\mathbf{R}_{nm}$ for different kinds of $(n,m)$ in the above. Combining them with \eqref{equ:equ_8}, we are able to obtain the matrix $\mathbf{D}$ and further evaluate $\Delta(\mathbf{w}_{\mathrm{opt}},\bm{\delta})$ according to \eqref{equ:prop_1.2}. Now we conclude the proof.
\subsection{Proof of Corollary \ref{cor:cor_1}}
\label{subsect:corollary_1}
When $\bm{\delta}=\mathbf{1}_Q$, we have $\mathbf{g}=\mathcal{E}_{\mathrm{s}}\bm{\lambda}^*$, where $\bm{\lambda}$ is defined as $\bm{\lambda}\triangleq[\bm{\lambda}_1^t,...,\bm{\lambda}_N^t]^t$. Meanwhile, $\mathbf{R}_{nm}=\mathbf{Y}_{nm}$, for any $n,m\in\mathbb{N}$. As a result,
\begin{equation}
\mathbf{F}\mathbf{R}_{nm}\mathbf{F}^{\dag}=
\begin{cases}
\mathbf{I}_Q+\mathcal{E}_{\mathrm{s}}\bm{\Lambda}_n\bm{\Lambda}_n^{\dag},\ \ &n=m,\\
\mathcal{E}_{\mathrm{s}}\bm{\Lambda}_n\bm{\Lambda}_m^{\dag},\ \ &n\neq m,
\end{cases}
\label{equ:equ_10}
\end{equation}
are all diagonal matrices, thus making $\mathbf{D}_{nm}$ exactly equal to $\mathbf{F}\mathbf{R}_{nm}\mathbf{F}^{\dag}$, for any $n,m\in\mathbb{N}$. Letting $\bm{\Lambda}\triangleq[\bm{\Lambda}_1,...,\bm{\Lambda}_N]$, we have $\mathbf{D}$ in this situation given as
\begin{equation}
\mathbf{D}=\mathbf{I}_{NQ}+\mathcal{E}_{\mathrm{s}}\bm{\Lambda}^{\dag}\bm{\Lambda}.
\end{equation}
Then, exploiting Woodbury formula \cite{horn2012matrix}, we get its inversion as follows
\begin{equation}
\mathbf{D}^{-1}=\mathbf{I}_{NQ}-\mathcal{E}_{\mathrm{s}}\bm{\Lambda}^{\dag}\left(\mathbf{I}_Q+\mathcal{E}_{\mathrm{s}}\bm{\Lambda}\bm{\Lambda}^{\dag}\right)^{-1}\bm{\Lambda}.
\end{equation}
We notice that $(\mathbf{I}_Q+\mathcal{E}_{\mathrm{s}}\bm{\Lambda}\bm{\Lambda}^{\dag})^{-1}$ is in fact a diagonal matrix; that is $\left(\mathbf{I}_Q+\mathcal{E}_{\mathrm{s}}\bm{\Lambda}\bm{\Lambda}^{\dag}\right)^{-1}=$
\begin{equation}
\begin{pmatrix}
\frac{1}{1+\mathcal{E}_{\mathrm{s}}\sum_{n=1}^{N}|\lambda_{n1}|^2} & 0 & \cdots & 0 \\
0 & \frac{1}{1+\mathcal{E}_{\mathrm{s}}\sum_{n=1}^{N}|\lambda_{n2}|^2} & \cdots & 0 \\
\vdots & \vdots & \ddots & \vdots \\
0 & 0 & \cdots & \frac{1}{1+\mathcal{E}_{\mathrm{s}}\sum_{n=1}^{N}|\lambda_{nQ}|^2}
\end{pmatrix}.
\end{equation}
Further, it is easy to verify that $\bm{\Lambda}\bm{\lambda}^*$ satisfies
\begin{equation}
\bm{\Lambda}\bm{\lambda}^*=\left[\sum_{n=1}^N|\lambda_{n1}|^2,...,\sum_{n=1}^N|\lambda_{nQ}|^2\right]^t.
\end{equation}
Then with all the above results and some further manipulations, we arrive at
\begin{equation}
\Delta(\mathbf{w}_{\mathrm{opt}},\bm{\delta})=\frac{1}{Q\mathcal{E}_{\mathrm{s}}}\mathbf{g}^{\dag}\mathbf{D}^{-1}\mathbf{g}
=\frac{1}{Q}\sum_{q=1}^Q\frac{\mathcal{E}_{\mathrm{s}}\sum_{n=1}^N|\lambda_{nq}|^2}{1+\mathcal{E}_{\mathrm{s}}\sum_{n=1}^N|\lambda_{nq}|^2},
\end{equation}
and now, it is straightforward to verify \eqref{equ:equ_7}.
\subsection{Proof of Corollary \ref{cor:cor_2}}
\label{subsect:corollary_2}
When $T=1$, we simply use $h_n$ to denote the channel coefficient corresponding to the $n$-th BS antenna. In this situation, the circulant matrix $\mathbf{C}_n$ reduces to a scaled identity matrix $h_n\mathbf{I}_Q$ and the diagonal matrix $\bm{\Lambda}_n$ turns out to be $\bm{\Lambda}_n=\mathbf{F}\mathbf{C}_n\mathbf{F}^{\dag}=h_n\mathbf{I}_Q$. As a result, $\mathbf{g}_n$ in \eqref{equ:prop_1.3} becomes
\begin{equation}
\mathbf{g}_n=h_n^*\mathcal{E}_{\mathrm{s}}\left[\delta_n+\bar{\delta}_n\cdot\sqrt{\frac{2}{\pi(|h_n|^2\mathcal{E}_{\mathrm{s}}+1)}}\right]\mathbf{1}_Q.
\end{equation}
If we let $\bm{\nu}\in\mathbb{C}^{N\times 1}$ collect the coefficients before $\mathbf{1}_Q$, i.e.,
\begin{equation}
(\bm{\nu})_n=h_n^*\mathcal{E}_{\mathrm{s}}\left[\delta_n+\bar{\delta}_n\cdot\sqrt{\frac{2}{\pi(|h_n|^2\mathcal{E}_{\mathrm{s}}+1)}}\right],
\end{equation}
then we have
\begin{equation}
\mathbf{g}=\bm{\nu}\otimes\mathbf{1}_Q,
\end{equation}
where $\otimes$ denotes right Kronecker product.

As for the matrix $\mathbf{D}$, with patient examination we find out that each of its blocks, $\mathbf{D}_{nm}$, is also a scaled identity matrix, for any $n,m\in\mathbb{N}$. Then letting $\mathbf{E}\in\mathbb{C}^{N\times N}$ collect the scaling factors before $\mathbf{I}_Q$, we have
\begin{equation}
\mathbf{D}=\mathbf{E}\otimes\mathbf{I}_Q,
\end{equation}
in which $\mathbf{E}$ is given as (with proof omitted) $(\mathbf{E})_{nm}=$
\begin{equation}
\begin{cases}
1+\delta_n\cdot|h_n|^2\mathcal{E}_{\mathrm{s}}, & \text{if}\ n=m,\\
h_n^* h_m\mathcal{E}_\mathrm{s}\bigg[\delta_n\delta_m+
\delta_n\bar{\delta}_m\cdot\sqrt{\frac{2}{\pi(|h_m|^2\mathcal{E}_\mathrm{s}+1)}}+\\
\ \ \ \ \ \ \ \ \ \ \bar{\delta}_n\delta_m\cdot\sqrt{\frac{2}{\pi(|h_n|^2\mathcal{E}_\mathrm{s}+1)}}\bigg]+\\
\bar{\delta}_n\bar{\delta}_m\!\cdot\!\frac{2}{\pi}\bigg[
\mathrm{arcsin}\Big(\frac{(h_n^*h_m)_{\mathrm{R}}\mathcal{E}_\mathrm{s}}
{\sqrt{|h_n|^2\mathcal{E}_\mathrm{s}+1}\sqrt{|h_m|^2\mathcal{E}_\mathrm{s}+1}}\Big)+\\
\ \ \ \ \ \ \ \ \ \ j\mathrm{arcsin}\Big(\frac{(h_n^*h_m)_{\mathrm{I}}\mathcal{E}_\mathrm{s}}
{\sqrt{|h_n|^2\mathcal{E}_\mathrm{s}+1}\sqrt{|h_m|^2\mathcal{E}_\mathrm{s}+1}}\Big)
\bigg], & \text{if}\ n\neq m.
\end{cases}
\end{equation}
Comparing $\bm{\nu}^*$ and $\mathbf{E}^*$ with \cite[Equ. (13) (14)]{liang2015mixed}, we notice that they are virtually the same except for some little differences due to the different scaling parameters of $\mathrm{sgn}(x)$.

We proceed by evaluating $\Delta(\mathbf{w}_{\mathrm{opt}},\bm{\delta})$ in this situation; that is
\begin{eqnarray}
\Delta(\mathbf{w}_{\mathrm{opt}},\bm{\delta})&=&\frac{1}{Q\mathcal{E}_{\mathrm{s}}}(\bm{\nu}\otimes\mathbf{1}_Q)^{\dag}(\mathbf{E}\otimes\mathbf{I}_Q)^{-1}(\bm{\nu}\otimes\mathbf{1}_Q)\nonumber\\
&=&\frac{1}{Q\mathcal{E}_{\mathrm{s}}}(\bm{\nu}^{\dag}\otimes\mathbf{1}_Q^{\dag})(\mathbf{E}^{-1}\otimes\mathbf{I}_Q)(\bm{\nu}\otimes\mathbf{1}_Q)\nonumber\\
&=&\frac{1}{Q\mathcal{E}_{\mathrm{s}}}(\bm{\nu}^{\dag}\mathbf{E}^{-1}\bm{\nu})\otimes(\mathbf{1}_Q^{\dag}\mathbf{I}_Q\mathbf{1}_Q)\nonumber\\
&=&\frac{1}{\mathcal{E}_{\mathrm{s}}}\bm{\nu}^{\dag}\mathbf{E}^{-1}\bm{\nu}.
\end{eqnarray}
Then we immediately find out that it is the same as that we obtained for frequency-flat SIMO channels in \cite[Prop. 3]{liang2015mixed}, and thus conclude the proof.
\subsection{Proof of Corollary \ref{cor:cor_3}}
\label{subsect:corollary_3}
Letting $\mathcal{E}_{\mathrm{s}}$ tend to zero, we have
\begin{equation}
\lim_{\mathcal{E}_{\mathrm{s}}\rightarrow 0}\frac{\mathbf{g}_n}{\mathcal{E}_{\mathrm{s}}}=\left(\delta_n+\bar{\delta}_n\cdot\frac{2}{\pi}\right)\bm{\lambda}_n^*.
\end{equation}

To simplify the invertible block matrix $\mathbf{D}$, we need to examine each of its blocks. First let us look at an arbitrary nondiagonal block, i.e., $\mathbf{D}_{nm}$ with $n\neq m\in\mathbb{N}$. From \eqref{equ:equ_9}, we observe that $\lim_{\mathcal{E}_{\mathrm{s}}\rightarrow 0} \mathbf{Y}_{nm}=\mathbf{O}_{Q}$, for any $n\neq m$, and on the other hand, $\lim_{\mathcal{E}_{\mathrm{s}}\rightarrow 0} \mathbf{Y}_{nn}=\mathbf{I}_Q$. As a result, $\mathbf{R}_{nm}$ for $n\neq m$ always approaches a zero matrix no matter which case it falls into, and consequently $\mathbf{F}^{\dag}\mathbf{R}_{nm}\mathbf{F}$ tends to be a zero matrix as well, since the unitary transformation $\mathbf{F}$ does not change the Frobenius norm of a matrix. In summary,
\begin{equation}
\lim_{\mathcal{E}_{\mathrm{s}}\rightarrow 0}\mathbf{D}_{nm}=\mathbf{O}_{Q},\ \ \forall\ n\neq m\in\mathbb{N}.
\end{equation}

For the diagonal blocks, if $n\in\mathbb{N}_{\mathrm{1}}$, we have $\mathbf{R}_{nn}=\mathbf{Y}_{nn}$, and from \eqref{equ:equ_10} it is obvious that
$\lim_{\mathcal{E}_{\mathrm{s}}\rightarrow 0}\mathbf{F}^{\dag}\mathbf{R}_{nn}\mathbf{F}=\mathbf{I}_Q$.
In other word, we have
$\lim_{\mathcal{E}_{\mathrm{s}}\rightarrow 0}\mathbf{D}_{nn}=\mathbf{I}_Q$, for any $n\in\mathbb{N}_{\mathrm{1}}$.
If $n\in\mathbb{N}_{\mathrm{0}}$, on the other hand, from \eqref{equ:equ_9} and \eqref{equ:equ_6} we obtain
$\lim_{\mathcal{E}_{\mathrm{s}}\rightarrow 0}\bm{\Theta}_{nn}=\mathbf{I}_Q$.
Then, applying \eqref{equ:equ_11} it is straightforward to verify that
$\lim_{\mathcal{E}_{\mathrm{s}}\rightarrow 0}\mathbf{R}_{nn}=\mathbf{I}_Q$.
Again, we have $\lim_{\mathcal{E}_{\mathrm{s}}\rightarrow 0}\mathbf{D}_{nn}=\mathbf{I}_Q$, for any $n\in\mathbb{N}_{\mathrm{0}}$. In summary,
\begin{equation}
\lim_{\mathcal{E}_{\mathrm{s}}\rightarrow 0}\mathbf{D}=\mathbf{I}_{NQ}.
\end{equation}

With all the above results, we have
\begin{eqnarray}
&&\lim_{\mathcal{E}_{\mathrm{s}}\rightarrow 0}\frac{\Delta(\mathbf{w}_{\mathrm{opt}},\bm{\delta})}{\mathcal{E}_{\mathrm{s}}}\nonumber\\
&=&\frac{1}{Q}\lim_{\mathcal{E}_{\mathrm{s}}\rightarrow 0}
\left(\frac{\mathbf{g}}{\mathcal{E}_{\mathrm{s}}}\right)^{\dag}
\mathbf{D}^{-1}
\left(\frac{\mathbf{g}}{\mathcal{E}_{\mathrm{s}}}\right)\nonumber\\
&\overset{(a)}{=}&\frac{1}{Q}\left(\lim_{\mathcal{E}_{\mathrm{s}}\rightarrow 0}\frac{\mathbf{g}}{\mathcal{E}_{\mathrm{s}}}\right)^{\dag}
\left(\lim_{\mathcal{E}_{\mathrm{s}}\rightarrow 0}\mathbf{D}^{-1}\right)
\left(\lim_{\mathcal{E}_{\mathrm{s}}\rightarrow 0}\frac{\mathbf{g}}{\mathcal{E}_{\mathrm{s}}}\right)\nonumber\\
&\overset{(b)}{=}&\frac{1}{Q}\left(\lim_{\mathcal{E}_{\mathrm{s}}\rightarrow 0}\frac{\mathbf{g}}{\mathcal{E}_{\mathrm{s}}}\right)^{\dag}
\left(\lim_{\mathcal{E}_{\mathrm{s}}\rightarrow 0}\mathbf{D}\right)^{-1}
\left(\lim_{\mathcal{E}_{\mathrm{s}}\rightarrow 0}\frac{\mathbf{g}}{\mathcal{E}_{\mathrm{s}}}\right)\nonumber\\
&=&\frac{1}{Q}\sum_{n=1}^N\left(\delta_n+\bar{\delta}_n\cdot\frac{2}{\pi}\right)\|\bm{\lambda}_n\|^2,
\end{eqnarray}
where (a) is obtained by applying the algebraic limit theorem since the limits of $\mathbf{g}/\mathcal{E}_{\mathrm{s}}$ and $\mathbf{D}^{-1}$ exist, and (b) comes from the fact that the inverse of a nonsingular matrix is a continuous function of the elements of the matrix, i.e., $\lim_{\mathcal{E}_{\mathrm{s}}\rightarrow 0}\mathbf{D}^{-1}=(\lim_{\mathcal{E}_{\mathrm{s}}\rightarrow 0}\mathbf{D})^{-1}$ \cite{stewart1969continuity}. Noting that $\log(1+x/(1-x))=x+o(x)$, as $x\rightarrow 0$, we immediately have \eqref{equ:equ_12}.
\subsection{Proof of Corollary \ref{cor:cor_4}}
\label{subsect:corollary_4}
When $\bm{\delta}=\mathbf{0}$ and as $\mathcal{E}_{\mathrm{s}}$ grows without bound, we have
\begin{equation}
\lim_{\mathcal{E}_{\mathrm{s}}\rightarrow\infty}\frac{\mathbf{g}_n}{\sqrt{\mathcal{E}_{\mathrm{s}}}}=\sqrt{\frac{2Q}{\pi}}\frac{\bm{\lambda}_n^*}{\|\bm{\lambda}_n\|}.
\end{equation}
For expositional concision, we denote the normalization of $\bm{\lambda}_n$ by $\bar{\bm{\lambda}}_n\triangleq\bm{\lambda}_n/\|\bm{\lambda}_n\|$, and accordingly define $\bar{\bm{\lambda}}\triangleq[\bar{\bm{\lambda}}_1^t,...,\bar{\bm{\lambda}}_N^t]^t$. Then $\mathbf{g}/\sqrt{\mathcal{E}_{\mathrm{s}}}$ in this situation approaches
\begin{equation}
\lim_{\mathcal{E}_{\mathrm{s}}\rightarrow\infty}\frac{\mathbf{g}}{\sqrt{\mathcal{E}_{\mathrm{s}}}}=\sqrt{\frac{2Q}{\pi}}\bar{\bm{\lambda}}^*,
\end{equation}
which is independent of $\mathcal{E}_{\mathrm{s}}$.

On the other hand, as $\mathcal{E}_{\mathrm{s}}$ tends to infinity we have
\begin{equation}
\lim_{\mathcal{E}_{\mathrm{s}}\rightarrow\infty}\bm{\Theta}_{nm}=Q\mathbf{F}^{\dag}\frac{\bm{\Lambda}_n}{\|\bm{\lambda}_n\|}\frac{\bm{\Lambda}_m^{\dag}}{\|\bm{\lambda}_m\|}\mathbf{F},
\label{equ:equ_13}
\end{equation}
which is independent of $\mathcal{E}_{\mathrm{s}}$ as well. Since when $\bm{\delta}=\mathbf{0}$ and as $\mathcal{E}_{\mathrm{s}}\rightarrow\infty$, $\mathbf{D}_{nm}$ is given by the combination of \eqref{equ:equ_8}, \eqref{equ:equ_11} and \eqref{equ:equ_13}, we conclude that $\lim_{\mathcal{E}_{\mathrm{s}}\rightarrow\infty}\mathbf{D}_{nm}$ exists for any $n,m\in\mathbb{N}$. If we define $\bar{\mathbf{D}}\triangleq\lim_{\mathcal{E}_{\mathrm{s}}\rightarrow\infty}\mathbf{D}$, then $\bar{\mathbf{D}}$ also exists and is independent of $\mathcal{E}_{\mathrm{s}}$. As a result, we have
\begin{eqnarray}
&&\lim_{\mathcal{E}_{\mathrm{s}}\rightarrow\infty}\Delta(\mathbf{w}_{\mathrm{opt}},\bm{\delta})\nonumber\\
&=&\frac{1}{Q}\lim_{\mathcal{E}_{\mathrm{s}}\rightarrow\infty}\left(\frac{\mathbf{g}}{\sqrt{\mathcal{E}_{\mathrm{s}}}}\right)^{\dag}\mathbf{D}^{-1}\left(\frac{\mathbf{g}}{\sqrt{\mathcal{E}_{\mathrm{s}}}}\right)\nonumber\\
&\overset{(a)}{=}&\frac{1}{Q}\left(\lim_{\mathcal{E}_{\mathrm{s}}\rightarrow\infty}\frac{\mathbf{g}}{\sqrt{\mathcal{E}_{\mathrm{s}}}}\right)^{\dag}\left(\lim_{\mathcal{E}_{\mathrm{s}}\rightarrow\infty}\mathbf{D}^{-1}\right)\left(\lim_{\mathcal{E}_{\mathrm{s}}\rightarrow\infty}\frac{\mathbf{g}}{\sqrt{\mathcal{E}_{\mathrm{s}}}}\right)\nonumber\\
&\overset{(b)}{=}&\frac{1}{Q}\left(\lim_{\mathcal{E}_{\mathrm{s}}\rightarrow\infty}\frac{\mathbf{g}}{\sqrt{\mathcal{E}_{\mathrm{s}}}}\right)^{\dag}\left(\lim_{\mathcal{E}_{\mathrm{s}}\rightarrow\infty}\mathbf{D}\right)^{-1}\left(\lim_{\mathcal{E}_{\mathrm{s}}\rightarrow\infty}\frac{\mathbf{g}}{\sqrt{\mathcal{E}_{\mathrm{s}}}}\right)\nonumber\\
&=&\frac{2}{\pi}\bar{\bm{\lambda}}^t\bar{\mathbf{D}}^{-1}\bar{\bm{\lambda}}^*,
\end{eqnarray}
where (a) is obtained by applying the algebraic limit theorem since both $\lim_{\mathcal{E}_{\mathrm{s}}\rightarrow\infty}\mathbf{g}/\sqrt{\mathcal{E}_{\mathrm{s}}}$ and $\lim_{\mathcal{E}_{\mathrm{s}}\rightarrow\infty}\mathbf{D}^{-1}$ exist, and (b) comes from the fact that the inverse of a nonsingular matrix is a continuous function of the elements of the matrix, i.e., $\lim_{\mathcal{E}_{\mathrm{s}}\rightarrow\infty}\mathbf{D}^{-1}=(\lim_{\mathcal{E}_{\mathrm{s}}\rightarrow\infty}\mathbf{D})^{-1}$ \cite{stewart1969continuity}.

\end{document}